%% file: manuscript_v18_arxiv.tex
\documentclass[journal,letterpaper,twocolumn,final,10pt]{IEEEtran}

\usepackage{amsthm}
\usepackage{amsmath}
\usepackage{amsfonts}
\usepackage{dsfont}
\usepackage{mathrsfs}
\usepackage{pifont}
\usepackage{amssymb}
\usepackage{verbatim}
\usepackage{upgreek}
\usepackage{color}

\usepackage{footmisc}
\usepackage[final]{graphicx}

\usepackage{subcaption}

\usepackage{algorithmic}
\usepackage{algorithm}
\usepackage{epsfig}
\usepackage{eucal}
\usepackage{cite}

\usepackage{balance}
\usepackage{setspace}
\usepackage{siunitx}
\usepackage{gensymb}
\usepackage{enumerate}

\usepackage{array}
\usepackage{textcomp}
\usepackage{stfloats}
\usepackage{url}
\usepackage{multirow}
\usepackage{booktabs}

\input{Defs2}

\begin{document}

\title{Subspace-based Second-Order Detection Architectures for Radar Signals
in the Presence of Unstructured Interference}

\author{Nicola Acito, \IEEEmembership{Member, IEEE}, Danilo Orlando, \IEEEmembership{Senior Member, IEEE},
and Giuseppe Ricci, \IEEEmembership{Senior Member, IEEE}
\thanks{The work of N. Acito and D. Orlando was supported in part by the Italian Ministry of 
Education and Research (MUR) in the framework of the FoReLab project (Departments of Excellence).}
\thanks{Nicola Acito and Danilo Orlando are with Department of Information Engineering, University of Pisa,
Via Caruso 16, 56122 Pisa, Italy. E-mail: {\tt nicola.acito@unipi.it; danilo.orlando@unipi.it}.}  
\thanks{Giuseppe Ricci is with Dipartimento di Ingegneria dell'Innovazione, Università del Salento,
Via Monteroni snc, 73100 Lecce, Italy. E-mail: {\tt giuseppe.ricci@unisalento.it}.}        
}

\maketitle

\begin{abstract}
This paper proposes a subspace-based decision scheme to detect  range-spread targets  
in presence of thermal noise, clutter, and intentional interference.
The target and intentional interference (i.e., noise-like jammer) effects are confined to the covariance matrix of the received
primary data vectors. The jammer is also present
in the secondary data vectors necessary to estimate the clutter covariance matrix.
Finally, an additional set of training
data, acquired by the system in passive mode and, hence, not contaminated by 
target and clutter echoes,  is used to avoid the  jammer parameters are unidentifiable. 
The detection problem is addressed by deriving an architecture based upon the likelihood ratio test
where the unknown second-order parameters are suitably estimated.
To this end, a unified theoretical framework is defined also providing
a geometrical interpretation of the estimation procedures.
The performance assessment, conducted also in comparison to a natural competitor,  highlights the effectiveness of the proposed
solution.

\end{abstract}


\section{Introduction}
\label{Sec:Introduction}
The subspace model for target and interference has met a large approval in the context
of adaptive radar detection over the last decades as corroborated by
the unlimited literature record (see, for example, \cite{kelly-tr,matchedSubspace,adaptiveSubspaceDetectors,
ScharfACE,deMaioSubspaceInvariance,BoseStein,BoseStein2,DD,SD,Ciuonzo1,
9591392,10742480,11558562,1396421,9887896}). This model assumes that
the signal of interest visits either a known or an unknown subspace of the observables. 
The motivation behind this wide success is due to the fact that in practice
target signature cannot be perfectly known and an uncertainty arises. 
The causes of this uncertainty are manifold and range
from hardware limitations to miscalibration errors and unbalanced channels. The subspace
model makes it possible to capture most of the energy loss that occurs 
when the nominal and the actual signatures of the
target are not perfectly aligned thereby mitigating 
the performance degradation of the conventional approaches.
As a matter of fact, most of conventional decision schemes 
compare a measure of the received energy along
the presumed target signature with a threshold or with an estimate of 
the entire received energy. This comparison
can be accomplished in the original observation space or 
after whitening data by means of an estimate of the disturbance
covariance matrix. It is clear that, for 
high signal-to-interference-plus-noise ratio values, when 
the actual target signature is different from 
the nominal one, the norm of
the component of the received data along the nominal 
target direction is lower than the norm
of the component along the actual target direction.

As stated before, in the open literature, there exist many excellent contributions on subspace detection. 
However, in what follows, for brevity, we review some 
of them with the aim of introducing and motivating 
the approach proposed in this paper and we apologize if we 
failed to mention any other important contribution.
To begin with, we recall the seminal subspace-based 
framework developed in \cite{kelly-tr}, where
the so-called generalized multivariate analysis of 
variance (GMANOVA) is applied to adaptive radar detection
by extending it to complex quantities. This model accounts for multiple snapshots 
and in addition to the subspace for the target steering vector, 
a subspace for target responses is considered.
Observe that many conventional detection problems can be 
considered as special cases of the GMANOVA model.
In \cite{matchedSubspace}, a general class of problems concerning the detection
of signals belonging to a known subspace is formulated 
and solved through the generalized likelihood ratio test (GLRT).
The interference is assumed to contribute not only to 
the second-order statistics but also 
to the first-order statistics of the received signal. 
Detectors derived in \cite{matchedSubspace} are adapted
to unknown covariance matrices in \cite{adaptiveSubspaceDetectors} including as special cases
well-known decision schemes such as Kelly's GLRT \cite{kelly1986adaptive} and the adaptive 
cosine/coherence estimator \cite{381910,ScharfACE}.
It is worth pointing out that signal model formulated in \cite{matchedSubspace,adaptiveSubspaceDetectors}
comprises a component (first-order parameter) that, in the radar context, is 
representative of structured interference.
The same problem is solved in \cite{deMaioSubspaceInvariance} by using the invariance principle
and a maximal invariant statistic, for a group of transformations that leaves invariant
the hypothesis test, is determined. Then, a class of invariant decision schemes are devised
by exploiting the distribution of the maximal invariant statistic 
and the likelihood ratio test (LRT).
The principle of invariance is also used in \cite{BoseStein} to address the conventional 
detection problem under the GMANOVA model.
Remarkably, it is shown that in the specific case of $1$-dimensional 
signal subspace and assuming that target contaminates
only one snapshot, Kelly's GLRT, the adaptive matched filter \cite{robey1992cfar}, and so-called loss factor 
\cite{BOR-Morgan} are functions of the maximal invariant statistic (see also \cite{BoseStein2}).

The first extension of the subspace approach to the case of range-spread targets and
data contaminated by structured interference can be found in \cite{DD,SD}.
The difference between these two contributions relies in the 
behavior of the target signature
over the range bins. In fact, in \cite{DD}, target components 
in each range bin are characterized by the same 
steering vectors, whereas in \cite{SD}, this constraint 
is removed and the signatures
might change over the range bins under test. In addition, the authors assume that
the structured 
subspace interference does not affect secondary data, 
namely, the set of data assumed free of useful signal components
and used to estimate the clutter spectral properties. Such a 
model fits scenarios where
the coherent jammers are used to deceive the victim radar by 
generating false targets that 
tie up the system's resources with the risk of wasting assets of high strategic value
(notice that in these papers 
the interference is also present under the null hypothesis). 
In this situation, a prospective genuine target might 
be missed \cite{EW101}.

A general framework that comprises the previously described scenarios as special 
cases is presented in \cite{Ciuonzo1},
where a generalization of the GMANOVA is studied and the maximal 
invariant statistic is obtained
with respect to a desirable group of transformations. 
Other important approaches to subspace detection 
can be found in \cite{9591392,10742480,11558562}.
In \cite{9591392}, the signature of the coherent interferer is assumed orthogonal
to the subspace of the signal, whereas in \cite{10742480}, 
a multiple hypothesis test is considered
where each hypothesis corresponds to a given subspace size.
This detection problem is solved by means
of the so-called penalized likelihood ratio test defined in \cite{Pia2021_1}. Finally,
in \cite{11558562}, the GMANOVA model is used to formulate a detection problem
for partially homogeneous scenarios, namely when the clutter 
in the data under test (primary data)
exhibits a different power with respect to the clutter in the 
secondary data. This power difference
is modeled in terms of a different scaling factor between the 
clutter covariance matrix in the primary
data and that in the secondary data.

So far, both useful and possible 
interfering signals are modeled
in terms of the first-order statistics. Such a model is reasonable
for coherent jammers \cite{ABORT_ECM}, but it is not strictly appropriate
in the case of noise-like jammers, whose components emulate the noise fluctuations
and, hence, cannot be considered deterministic. Thus, an approach based on second-order statistics
is required. In this respect, a first
formulation of the subspace detection problem in terms of second order statistics can be found
in \cite{1396421}, where primary data are initially whitened with the 
sample covariance matrix of the secondary data and then
the GLRT is derived. In \cite{9887896}, a comprehensive framework for subspace detection
that includes first-order and second-order subspace models is derived. Moreover, the case where
the signal subspace is unknown is addressed.
However, such second-order contributions do not 
incorporate the intentional interference into the 
signal model at the design stage.
Thus, it become quite evident that a 
subspace-based second-order formulation that includes 
both signal of interest (representative of a fluctuating target \cite{Richards_BP})
and interference (representative of noise-like jammers) is missing in the open literature.

This paper represents an initial attempt to fill the aforementioned 
gap by proposing a subspace-based detector
within the second-order model for both the signal of interest and the intentional interference.
Focusing on range-spread targets, we formulate the detection problem at hand
in terms of a binary hypothesis test
where signal and interference components are confined to the covariance matrix of the received
data vectors. More importantly, for each component
of the overall covariance matrix, we exploit the structure that arises from the
design assumptions. Given the complexity of the problem and to solve
an identifiability issue, we resort to an additional set of training
data acquired by the system in passive mode. As a consequence, such data are not contaminated by 
target and clutter echoes; instead, they are representative of the ubiquitous
thermal noise and of a possible noise-like jammer \cite{DoppioTraining,9444189}.
It is important to point out that the effects of the noise-like jammer are present
not only in the primary data but also in the secondary data due
to the continuous transmission of the wideband noisy signals \cite{antennaBased,stimson2014stimson}.

The proposed formulation concerns the discrimination between the null hypothesis, where
intentional interference, clutter, and noise components are present in data, against the alternative
hypothesis that data also include target echoes. The motivation behind this choice
is that, when the search function is active, the main beam
can be steered along directions far from the target. In this situation, data
are affected by thermal noise, clutter, and noise-like jammer, because the latter
continuously illuminates the victim radar. 
This detection problem is addressed by deriving an architecture based upon the LRT
where the unknown second-order parameters are suitably estimated by means
of a ``sequential'' maximum likelihood approach (a point better explained
in Section \ref{Subsec:LRTdesign}). Interestingly, we provide theoretical
results, including the proof of the aforementioned 
identifiability issue and a geometrical interpretation of the second-order
estimates, that exhibit a more general value and can be applied to analogous
problems in different contexts. From an architectural point of view,
the proposed detector can be framed as follows. The system uses a preliminary
stage to detect the noise-like jammer \cite{Pia2021_2} and, when the latter
is not present, data are processed by a subsequent conventional detection stage.
In the opposite case, data are processed by a stage 
that implements the herein proposed detector.
Finally, the performance assessment highlights the effectiveness of the proposed
solution also in comparison with most natural existing 
competing technique developed in \cite{SD}.

The remainder of this paper is organized as follows. In the next section, we
provide the formal statement of the problem along with the definition
of data distributions and the related parameters. In Section \ref{Sec:Designs},
we provide some general theoretical results that are used to estimate
the unknown distribution parameters and derive the LRT-based detector.
Section \ref{Sec:Identifiability} contains the proof of the identifiability issue
related to some distribution parameters, while in Section \ref{Sec:PerformanceAssessment},
illustrative examples are presented and discussed. Finally, in Section \ref{Sec:Conclusions}
we draw the conclusions and outline future research tracks.

\subsection{Notation} 
In the sequel, vectors and matrices are denoted by boldface lower-case and upper-case letters, respectively. 
Symbols $\det(\cdot)$, $\tr(\cdot)$, $(\cdot)^T$, and  $(\cdot)^\dag$ denote the determinant, trace, 
transpose, and  conjugate transpose, respectively. 
$\C$ is the set 
of complex numbers, and $\C^{N\times M}$ is the Euclidean space of $(N\times M)$-dimensional complex matrices 
(or vectors if $M=1$).
%The imaginary unit is indicated by $j$.
The $(N\times N)$-dimensional diagonal matrix with entries $a_1,\ldots,a_N$
is denoted by $\diag(a_1,\ldots,a_N)$. Given $\bA\in\C^{N\times M}$, its Frobenius norm 
is $\|\bA\|_F$.
The $i$th entry of a vector $\bx$ is denoted by $\bx(i)$.
Symbols $\bI$ and $\bzero$ indicate the identity matrix and the null matrix or vector, 
respectively, whose size depends on the context.
%$\buno_N$ is a vector of size $N$ whose entries are equal to $1$.
For any Hermitian matrix $\bA$, $\bA\succ\bzero$ means that $\bA$ is a positive
definite matrix, while $\bA\succeq\bzero$ means that $\bA$ is a positive semidefinite matrix. 
The acronyms IID and PDF stand for independent and identically
distributed and probability density function, respectively.
Finally, 
if $\bX\in\C^{N\times K}$, 
$\bX\sim\cC\cN_{N,K}(\bm, \bSigma,\bI)$ means that the columns of $\bX$ are 
$K$ IID $N$-dimensional random vectors 
following the circularly invariant
Gaussian distribution with mean $\bm$ and covariance matrix $\bSigma\succ\bzero$.

\section{Problem Formulation}
\label{Sec:problemFormulation}
Let us consider a multichannel radar system with $N$ channels and 
denote by $\bZ_P=[\bz_{p,1},\ldots,\bz_{p,K_P}]\in\C^{N\times K_P}$ the
primary data matrix whose columns are statistically independent. Moreover, let us
assume that the system acquires a set of secondary data used to estimate
the unknown interference (thermal noise, clutter, and noise-like jamming signals) 
parameters and denote
by $\bZ_S=[\bz_{s,1},\ldots,\bz_{s,K_S}]\in\C^{N\times K_S}$ the respective data matrix.
%whose columns are statistically independent vectors. 
Suppose also that the columns of the  matrix $[ \bZ_P,\bZ_S]$ are independent random vectors.
As stated before, the interference sources are the environment reverberation, which generates
clutter echoes characterized by a positive semidefinite covariance matrix
$\bM_c\in\C^{N\times N}$, 
the electronic devices, which are responsible for thermal noise with
power $\sigma^2_n>0$,
and a noise-like jammer represented by a suitable covariance component that
will be described below. At the design stage, we assume that both target and jammer
steering vectors belong to known subspaces spanned by the columns of 
$\bH\in\C^{N\times r}$ and $\bJ\in\C^{N\times q}$, respectively, with $r+q\leq N$
and such that $[\bH,\bJ]\in\C^{N\times (r+q)}$ is a full-column rank matrix.
Let us recall here that the use of the subspace model is motivated by the possible uncertainty
associated with the beam pointing (due for instance to phase/amplitude miscalibration errors,
imperfections of the single channel beam, and so on). As for the a priori information
related to the jammer subspace, it can be provided by a preliminary stage that listens to the environment to 
detect possible active interferers and estimate their directions of arrival.
As explained in Section \ref{Sec:Introduction}, 
if the presence of a jammer is declared, a detection
stage is activated to discriminate the hypothesis that data contain clutter, jammer, and
thermal noise components from the hypothesis that data are also affected by target components.
Moreover, the objective of the noise-like jammer is 
to blind the victim radar (in the case the latter uses an adaptive threshold) or to generate
distributed clutter that masks the genuine target (when the system cannot
adapt to large disturbance variations). To this end, it continuously illuminates
the victim radar and, hence, its contribution 
affects all the collected data. 

With this said, we formulate
the detection problem at hand as follows
\be
\begin{cases}
H_0:
\begin{cases}
\bZ_P\sim\cC\cN_{N,K_P}(\bzero,\sigma^2_n\bI+\bM_c+\bJ\bQ\bJ^\dag,\bI),
\\
\bZ_S\sim\cC\cN_{N,K_S}(\bzero,\sigma^2_n\bI+\bM_c+\bJ\bQ\bJ^\dag,\bI),
\end{cases}
\\
\vspace{-3mm}
\\
H_1:
\begin{cases}
\bZ_P\sim\cC\cN_{N,K_P}(\bzero,\sigma^2_n\bI+\bM_c
\\
\hspace{24mm} +\bH\bP\bH^\dag+\bJ\bQ\bJ^\dag,\bI),
\\
\bZ_S\sim\cC\cN_{N,K_S}(\bzero,\sigma^2_n\bI+\bM_c+\bJ\bQ\bJ^\dag,\bI),
\end{cases}
\end{cases}
\label{eqn:test1}
\ee
where $\bP\in\C^{r\times r}$ and $\bQ\in\C^{q\times q}$ are unknown positive semidefinite matrices.
The other unknown parameters are $\sigma^2_n$ and $\bM_c$.

In addition to the conventional primary and secondary sets, the radar system
exploits another data set obtained by opening the receiver before transmitting
the operating waveform and beyond the instrumental range. 
The samples of this set are representative of thermal noise 
and jammer. 
In other words, denoting by $\bX=[\bx_1,\ldots,\bx_{K_J}]\in\C^{N\times K_J}$
with $K_J\ge q$ this data matrix (again, we suppose that the columns 
of the  matrix $[\bZ_P,\bZ_S,\bX]$ are independent random vectors),
we assume that
\be
\bX\sim\cC\cN_{N,K_J}(\bzero,\sigma^2_n\bI+\bJ\bQ\bJ^\dag).
\ee
It is important to observe that the above matrix can be used to estimate 
$\bQ$ and $\sigma^2_n$ under both hypotheses.

%%%%As an alternative formulation of the detection problem at hand, we could neglect 
%%%%the structure of the thermal noise covariance matrix and
%%%%incorporate it and clutter component into a unique covariance matrix denote by $\bM$. As a consequence,
%%%%problem \eqref{eqn:test1} can be recast as
%%%%\be
%%%%\begin{cases}
%%%%H_0:
%%%%\begin{cases}
%%%%\bZ_P\sim\cC\cN_{N,K_P}(\bzero,\bM+\bJ\bQ\bJ^\dag,\bI),
%%%%\\
%%%%\bZ_S\sim\cC\cN_{N,K_S}(\bzero,\bM+\bJ\bQ\bJ^\dag,\bI),
%%%%\end{cases}
%%%%\\
%%%%\vspace{-3mm}
%%%%\\
%%%%H_1:
%%%%\begin{cases}
%%%%\bZ_P\sim\cC\cN_{N,K_P}(\bzero,\bM+\bH\bP\bH^\dag+\bJ\bQ\bJ^\dag,\bI),
%%%%\\
%%%%\bZ_S\sim\cC\cN_{N,K_S}(\bzero,\bM+\bJ\bQ\bJ^\dag,\bI),
%%%%\end{cases}
%%%%\end{cases}
%%%%\label{eqn:test2}
%%%%\ee
%%%%where we assume that the eigenvalues of $\bM$ are bounded 
%%%%from below by a system parameter $\xi>0$, which accounts for the power 
%%%%of the thermal noise, namely $\bM - \xi\bI \succeq \bzero$.

Finally, we conclude this section by providing the expressions 
of the PDFs that will be used in the ensuing derivations.
Specifically, under $H_0$ of \eqref{eqn:test1}, 
the PDF of $\bZ_P$ is given by
\begin{multline}
f_0(\bZ_P;\sigma^2_n,\bM_c,\bQ)=\left[\frac{1}{\pi^N \det(\sigma_n^2\bI+\bM_c+\bJ\bQ\bJ^\dag)}\right]^{K_P}
\\
\times \exp\left\{-\tr\left[(\sigma_n^2\bI+\bM_c+\bJ\bQ\bJ^\dag)^{-1}\bZ_P\bZ_P^\dag\right]\right\},
\label{eqn:PDF_ZpZs}
\end{multline}
whereas under $H_1$ of \eqref{eqn:test1}, the PDF of $\bZ_P$ has the following expression
\begin{multline}
f_1(\bZ_P;\sigma^2_n,\bM_c, \bP, \bQ)
\\
=\left[\frac{1}{\pi^N \det(\sigma_n^2\bI+\bM_c+\bH\bP\bH^\dag+\bJ\bQ\bJ^\dag)}\right]^{K_P}
\\
\exp\left\{-\tr\left[(\sigma_n^2\bI+\bM_c+\bH\bP\bH^\dag+\bJ\bQ\bJ^\dag)^{-1}\bZ_P\bZ_P^\dag\right]\right\}.
\label{eqn:PDF_Zp_H1_3}
\end{multline}

The PDF of $\bZ_S$ for both hypotheses of \eqref{eqn:test1} is
\begin{multline}
f(\bZ_S;\sigma^2_n,\bM_c, \bQ)=\left[\frac{1}{\pi^N \det(\sigma_n^2\bI+\bM_c+\bJ\bQ\bJ^\dag)}\right]^{K_S}
\\
\exp\left\{-\tr\left[(\sigma_n^2\bI+\bM_c+\bJ\bQ\bJ^\dag)^{-1}\bZ_S\bZ_S^\dag\right]\right\}.
\label{eqn:PDF_Zs_H1_3}
\end{multline}

As for the PDF of $\bX$, it is given by
\begin{multline}
f(\bX;\sigma^2_n,\bQ)=\left[\frac{1}{\pi^N \det(\sigma^2_n\bI+\bJ\bQ\bJ^\dag)}\right]^{K_J}
\\
\times\exp\left\{-\tr\left[(\sigma^2_n\bI+\bJ\bQ\bJ^\dag)^{-1}\bX\bX^\dag\right]\right\}.
\label{eqn:PDF_X}
\end{multline}

\section{Detector Designs}
\label{Sec:Designs}
In this section, we design a decision scheme for problem \eqref{eqn:test1}
by exploiting the maximum likelihood approach. Specifically, when the maximization of
the joint likelihood of $\bZ_P$ and $\bZ_S$ with respect to the unknown parameters is mathematically involved, 
we consider subproblems that allow us to 
obtain suboptimal (in the sense of the maximum likelihood)
closed-form
expressions of the parameter estimates.
As a preliminary step towards the derivation of the decision scheme, 
we introduce some important technical results.

\subsection{Preliminary results}

First, we formulate the following important 
proposition that will come into play for the design of the detector.

\begin{prop}\label{prop:maxLogDetTr_C_PSD}
Let $\bC\in\C^{M\times M}$ and $\bR\in\C^{M\times M}$ be two Hermitian 
matrices such that $\bC\succeq \bzero$ and $\bR\succeq \bzero$.
Then, the maximizer of the objective function
\be
-\log\det(\bI+\bC)-\tr[(\bI+\bC)^{-1}\bR].
\label{eq:objective_function_prop1}
\ee
with respect to $\bC$ is given by
\be
\widehat{\bC}=\bU \diag(\max\{\lambda_1-1,0\},\ldots,\max\{\lambda_M-1,0\})\bU^\dag,
\ee
where $\lambda_1\ge\lambda_2\ge\ldots,\ge\lambda_M\ge 0$ are the eigenvalues of $\bR$ and
$\bU\in\C^{M\times M}$ is a unitary matrix whose columns are the eigenvectors of $\bR$.
\end{prop}

\begin{proof}
Let us rewrite the objective function as follows
\be
\log\det[(\bI+\bC)^{-1}]-\tr[(\bI+\bC)^{-1}\bR]
\label{eqn:objectiveFunction}
\ee
and notice that the determinant in the first term is independent of the eigenvectors of $\bC$.
As a matter of fact, by using the eigendecomposition of $\bC=\bV\bGamma\bV^\dag$,
with $\bGamma=\diag(\gamma_1,\ldots,\gamma_M)$ 
the matrix of its eigenvalues, $\gamma_1\ge\gamma_2\ge\ldots\ge\gamma_M\ge 0$, and $\bV\in\C^{M\times M}$ the unitary matrix of the
corresponding eigenvectors, the objective function can be written as
\be
\log\det[(\bI+\bGamma)^{-1}]-\tr[\bV(\bI+\bGamma)^{-1}\bV^\dag\bR].
\label{eqn:objectiveFunction1}
\ee
Now, by Theorem 4.3.53 of \cite{horn2012matrix}, we obtain that
\be
-\tr[\bV(\bI+\bGamma)^{-1}\bV^\dag\bR]\le -\sum_{m=1}^M \frac{\lambda_m}{1+	\gamma_m}
\ee
where $\lambda_1\ge\ldots\ge\lambda_M\ge 0$ are the eigenvalues of $\bR$ with the equality
if and only if matrices $\bV(\bI+\bGamma)^{-1}\bV^\dag$ and $\bR$ commute.
Therefore by Theorem 4.1.6 of \cite{horn2012matrix}, the upper bound is achieved if and only if
$\bV(\bI+\bGamma)^{-1} \bV^\dag$ and $\bR$ share the same eigenvectors
(namely, they are simultaneously diagonalizable). As a consequence, $\bV$ is selected to be
the unitary matrix of the eigenvectors of $\bR$ and the objective function can be recast as
\begin{align}
&\log\det[(\bI+\bGamma)^{-1}]-\sum_{m=1}^M \frac{\lambda_m}{1+\gamma_m}
\nonumber
\\
&=\sum_{m=1}^M \log \left(\frac{1}{1+\gamma_m}\right)-\sum_{m=1}^M \frac{\lambda_m}{1+\gamma_m}
\\
&=\sum_{m=1}^M g(\gamma_m;\lambda_m),
\label{eqn:objectiveFunction01}
\end{align}
where $g(\gamma_m;\lambda_m)=\log \left[1/(1+\gamma_m)\right]-\lambda_m/(1+\gamma_m)$.

Finally, we need to maximize the objective function with respect to the $\gamma_m$s with the constraint
$\gamma_m\ge 0$, $\forall m$ (since $\bC\succeq \bzero$). 
To this end, we observe that $g(\gamma_m;\lambda_m)$ is continuous
in $[0,+\infty)$ and
\begin{align}
&-\infty<\lim_{\gamma_m\rightarrow 0}g(\gamma_m;\lambda_m)=-\lambda_m \le 0,
\\
&\lim_{\gamma_m\rightarrow +\infty}g(\gamma_m;\lambda_m)=-\infty.
\end{align}
Moreover, the first derivative of $g(\gamma_m;\lambda_m)$ with 
respect to $\gamma_m$ is
\be
\frac{dg(\gamma_m;\lambda_m)}{d \gamma_m}=-\frac{1}{1+\gamma_m}+\frac{\lambda_m}{(1+\gamma_m)^2}.
\ee
It follows that
the sign of the derivative is always negative if $\lambda_m-1<0$ and, hence, 
$g(\gamma_m;\lambda_m)$ is decreasing. On the other hand, if $\lambda_m-1>0$, then
the derivative is positive when $\gamma_m<\lambda_m-1$ and negative when $\gamma_m>\lambda_m-1$.
It turns out that
\be
\gamma_m=\max\{\lambda_m-1,0\}.
\label{eqn:solutionGamma}
\ee

Summarizing, the estimate of $\bC$ is obtained by setting to zero the entries of $\bLambda-\bI$
that are negative, namely
\be
\widehat{\bC}=\bU
\begin{bmatrix}
\max\{\lambda_1-1,0\} & \ldots & 0
\\
\vdots & \ddots & \vdots
\\
0 & \ldots & \max\{\lambda_M-1,0\}
\end{bmatrix}
\bU^\dag.
\label{eqn:solutionCmatrix}
\ee
\end{proof}

\noindent
{\bf Remark 
(geometrical
interpretation of
Proposition \eqref{prop:maxLogDetTr_C_PSD}).}
Let us consider now, for $\bR \succ \bzero$,
the solution $\bar{\bC}$ of the following optimization problem
\be
\log\det[(\bI+\bC)^{-1}]-\tr[(\bI+\bC)^{-1}\bR]
\label{eqn:objectiveFunction_unconstrained}
\ee
with $\bC$ Hermitian and $(\bI+\bC)\succ \bzero$.
It is worth highlighting that \eqref{eqn:solutionCmatrix} can be viewed as the projection
of $\bar{\bC}$
onto the cone of 
the $(M\times M)$-dimensional positive semidefinite Hermitian
matrices. To prove this fact,
first observe that
\be
\bar{\bC}(\bR)=\bR-\bI.
\ee
In fact,
%the objective function \eqref{eqn:objectiveFunction_unconstrained}
%can be rewritten as
%$$
%\log\det[(\bI+\bC)^{-1}]-\tr[(\bI+\bC)^{-1}\bR].
%$$
the change of variable $\bX=(\bI+\bC)^{-1}\succ \bzero$ leads to
the function
\be
g(\bX)=\log\det(\bX)-\tr(\bX\bR)
\ee
that is concave \cite{boyd2004convex}.
In order to maximize $g(\bX)$ with respect to $\bX$, we proceed by setting 
to zero its first (complex) derivative with respect to $\bX$ \cite{hjørungnes2011complex}
to come up with
\be
\frac{\partial g(\bX)}{\partial \bX}=\bX^{-T}-\bR^T=\bzero
\Rightarrow \bX^{-1}=\bR.
\ee
It clearly follows that the unique value of $\bC$ that maximizes the objective function is
\be
\bar{\bC}(\bR)=\bR-\bI=\bU \left( \bLambda -\bI \right) \bU^\dag
\ee
with $\bU$ and $\bLambda$ the unitary matrix of the eigenvectors of $\bR$ and  the diagonal matrix of the corresponding eigenvalues, respectively.

Now, we recall that the projection of $\bar{\bC}$ is the positive semidefinite
Hermitian matrix $\bC$ that minimizes some distance with respect to $\bar{\bC}$.
For the case at hand, we consider the distance defined by the Frobenius norm, namely
\begin{align}
&\min_{\bC\succeq \bzero} \| \bC-\bar{\bC} \|_F^2 
\nonumber
\\
&=
\min_{\bC\succeq \bzero} \tr[(\bC-\bar{\bC})^\dag(\bC-\bar{\bC})]
\nonumber
\\
&=\min_{\bC\succeq \bzero} \tr[(\bU^\dag\bC\bU-\left(\bLambda-\bI\right))^\dag(\bU^\dag\bC\bU-\left(\bLambda-\bI\right))]
\nonumber
\\
&=\min_{\bB\succeq \bzero} \tr[(\bB-\left(\bLambda-\bI\right))^\dag(\bB-\left(\bLambda-\bI\right))]
\nonumber
\\
&=\min_{\bB\succeq \bzero} \|\bB-\left(\bLambda-\bI\right)\|_F^2
\nonumber
\\
&=\min_{\bB\succeq \bzero} \left[ \sum_{m=1}^M |b_{m,m}-(\lambda_{m}-1)|^2 
+\sum_{m,n=1\atop m\neq n}^M |b_{n,m}|^2 \right]
\end{align}
where $\bB=\bU^\dag\bC\bU$ is still positive semidefinite.
In order to accomplish the last minimization with respect to $\bB$ with the constraint of positive semidefiniteness,
we notice that the second positive term depends on $\bB$ only and can be set to zero. It follows
that $\bB$ must be a diagonal matrix implying that the projection and $\bar{\bC}$ share the same eigenvectors.
Then, $b_{m,m}$ must be an eigenvalue of $\bC$ and, since it is required that $\bB$ is positive semidefinite,
the first term is minimized when
\be
b_{m,m}=\max\{\lambda_m-1,0\},
\ee
that is exactly the expression obtained in \eqref{eqn:solutionGamma}.

Next result generalizes Proposition 1 by replacing $\bI$ with a positive definite matrix $\bA$.
\begin{corollary}\label{cor:maxLogDetTr_C_PSD}
Let $\bB\in\C^{M\times M}$ and $\bR\in\C^{M\times M}$ be two Hermitian 
matrices such that $\bB\succeq\bzero$ and $\bR\succeq\bzero$,
$\bA \in\C^{M\times M}$ be a Hermitian matrix such that $\bA\succ \bzero$,
and consider the following objective function
\[
-\log\det(\bA+\bB) \\
-\tr[(\bA+\bB)^{-1}\bR].
\]
The maximizer of the above function with respect to $\bB$ is given by
\[
\widehat{\bB}=
\bA^{1/2}
\bU \diag(\max\{\lambda_1-1,0\},\ldots,\max\{\lambda_M-1,0\})\bU^\dag \bA^{1/2},
\]
where $\lambda_1\ge\lambda_2\ge\ldots,\ge\lambda_M\ge 0$ are the eigenvalues of $\bA^{-1/2}\bR\bA^{-1/2}$ and
$\bU\in\C^{M\times M}$ is a unitary matrix whose columns are the eigenvectors of $\bA^{-1/2}\bR\bA^{-1/2}$.
\end{corollary}

\begin{proof}
Let us rewrite the objective function as follows
\begin{align}
\nonumber
&-\log\det(\bA) -\log\det(\bI+\bA^{-1/2}\bB\bA^{-1/2}) \\
&-\tr[(\bI+\bA^{-1/2}\bB\bA^{-1/2})^{-1}\bA^{-1/2}\bR\bA^{-1/2}].
\label{eq:objective_function_cor1}
\end{align}
Thus, we can apply Proposition \ref{prop:maxLogDetTr_C_PSD}
with $\bA^{-1/2}\bB\bA^{-1/2}\succeq\bzero$ and $\bA^{-1/2}\bR\bA^{-1/2}\succeq\bzero$
in place of $\bC$ and $\bR$, respectively.
\end{proof}

A related result is the following.
\begin{corollary}\label{cor:maxLogDetTrM_PD}
Let $\bC\in\C^{M\times M}$, $\bR\in\C^{M\times M}$, and $\bF\in\C^{M\times M}$ 
be three Hermitian matrices such that $\bC\succeq\bzero$,
$\bR\succeq\bzero$, and $\bF\succ\bzero$.
Moreover, we assume that $(\bF-\gamma\bI)\succeq\bzero$ with $\gamma>0$.
Then, the maximizer of the following objective function
\be
-\log\det(\bF+\bC)-\tr[(\bF+\bC)^{-1}\bR]
\ee
with respect to $\bF$ is given by
\begin{align}
&\widehat{\bF}=\gamma \bI + \left( \bC +\gamma \bI \right)^{1/2} \bU \diag\Big(\max\{\lambda_1-1,0\},
\nonumber
\\ 
& \ldots,\max\{\lambda_M-1,0\}\Big)\bU^\dag \left( \bC +\gamma \bI \right)^{1/2},
\end{align}
where $\lambda_1\ge\lambda_2\ge\ldots,\ge\lambda_M \ge0$
are the eigenvalues of the matrix
$\left( \bC +\gamma \bI \right)^{-1/2} \bR \left( \bC +\gamma \bI \right)^{-1/2}$ 
while $\bU\in\C^{M\times M}$ is the unitary matrix whose columns 
are the corresponding eigenvectors.
\end{corollary}

\begin{proof}
Let us rewrite the objective function as follows
\begin{multline}
-\log\det[(\bF-\gamma \bI) +\bC+\gamma \bI]
\\
-\tr\left\{\left[\left(\bF-\gamma \bI\right)+\bC+\gamma \bI\right]^{-1}\bR\right\}.
\end{multline}
We can now apply Corollary \ref{cor:maxLogDetTr_C_PSD} with 
$\bA=\bC+\gamma \bI\succ\bzero$ and $\bB=\bF-\gamma\bI\succeq\bzero$.
The result follows immediately.
\end{proof}

A final result that will be used in the derivations of the decision schemes is the following.
It is a more detailed reformulation of the results given in \cite{Bresler}.

\begin{prop}\label{stima_sigmaquadro}
Let $\phi_1, \ldots, \phi_d$, $2\le d<N$, be a sequence of nonnegative numbers ordered in nonincreasing order.
The function 
$$
g(\sigma^2_n)=g_1(\sigma^2_n)+g_2(\sigma^2_n)
$$
with
\[
g_1(\sigma^2_n) =
- \sum_{i=1}^d \left[
\log \left( \max \left\{ \frac{\phi_i}{\sigma^2_n}, 1 \right\} \right) 
+ \frac{\displaystyle \frac{\phi_i}{\sigma^2_n} }{ \displaystyle\max \left\{\frac{\phi_i}{\sigma^2_n}, 1 \right\}  } 
\right]
\]
and
\[
g_2(\sigma^2_n)=-N \log \sigma^2_n -\frac{a}{\sigma^2_n},
\]
with $a>0$, has a 
maximizer with respect to $\sigma^2_n >0$ given by the unique solution of the equation
\be
a+ \left( \sigma^2_n \right)^2 \frac{d g_1(\sigma^2_n)}{d \sigma^2_n}=N \sigma^2_n,
\label{eq:sigman-prop2}
\ee
given by eqs. \eqref{eq:interval-Bresler} and \eqref{eq:sigmaqj-Bresler}.
\end{prop}

\begin{proof}
Notice that $g(\cdot)$ is continuous in $(0,+\infty)$ and it is not difficult to show that
\[
\lim_{\sigma^2_n \rightarrow 0^+}  g(\sigma^2_n)=\lim_{\sigma^2_n \rightarrow +\infty}  g(\sigma^2_n)=-\infty.
\]
Thus, the maximizer can be obtained by setting to zero the first derivative 
of $g$ with respect to $\sigma^2_n$.
Now, computing the derivative of $g$ with respect to $\sigma^2_n$ leads to
\[
\frac{d g(\sigma^2_n)}{d \sigma^2_n}
= -\frac{N}{\sigma^2_n} 
+ \frac{a}{ \left(\sigma^2_n \right)^2} + \frac{dg_1(\sigma^2_n)}{d \sigma^2_n} = 0
\]
and eventually to \eqref{eq:sigman-prop2},
where
\[
\frac{dg_1(\sigma^2_n)}{d \sigma^2_n} =
\left\{
\begin{array}{lr}
\displaystyle d/\sigma^2_n, & 0 < \sigma^2_n \leq \phi_d, 
\\
\displaystyle \frac{(j-1)}{\sigma^2_n} 
+ \frac{1}{\left( \sigma^2_n \right)^2}\sum_{i=j}^d \phi_i,  & \phi_j < \sigma^2_n \leq \phi_{j-1}, 
\\
& j=d, \ldots, 2, 
\\
\vspace{-1mm}
\displaystyle \frac{1}{\left( \sigma^2_n \right)^2}\sum_{i=1}^d \phi_i,  & \phi_1 < \sigma^2_n.
\end{array}
\right.
\]
Although \eqref{eq:sigman-prop2} is nonlinear and
cannot be solved in closed form, we observe that
its left-hand side, namely
\begin{align*}
&h(\sigma^2_n)= a+\left( \sigma^2_n \right)^2 \frac{dg_1(\sigma^2_n)}{d \sigma^2_n} 
\\ &=
a+\left\{
\begin{array}{lr}
d \sigma^2_n, & \phi_{d+1}=0 \le \sigma^2_n \leq \phi_d, 
\\
\displaystyle (j-1)\sigma^2_n + \sum_{i=j}^d \phi_i,  & \phi_j < \sigma^2_n \leq \phi_{j-1}, 
\\
& j=d, \ldots, 2, 
\\
\displaystyle \sum_{i=1}^d \phi_i,  & \phi_1 < \sigma^2_n < \phi_0=+\infty,
\end{array}
\right.
\end{align*}
is piecewise linear and increasing in between $a$ and $a+\sum_{i=1}^d \phi_i$.
Moreover, it is a concave function since the slope of each segment decreases from $d$
to zero with increments equal to $1$ (corresponding to the bisector of the first quadrant). 
To prove concavity of $h(\sigma^2_n)$, we show that the chord with endpoints 
$(0,a)$ and $(\phi_1,\sum_{i=1}^d\phi_i)$, $c(\sigma^2_n)$, is always lower than
$h(\sigma^2_n)$ \cite{boyd2004convex}. As a matter of fact, observe that
\be
c(\sigma^2_n)=\sum_{i=1}^d \frac{\phi_i}{\phi_1} \sigma^2_n + a, \quad 0\le \sigma^2_n \le \phi_1.
\ee
Thus, when $0\le \sigma^2_n\le\phi_d$, we have that
\be
h(\sigma^2_n)-c(\sigma^2_n)=\left(d-\sum_{i=1}^d \frac{\phi_i}{\phi_1}\right)\sigma^2_n\ge 0,
\ee
where the inequality is due to the fact that $\forall i=1,\ldots,d: \ \phi_i/\phi_1 \le 1$.
Now, let us consider the case $\phi_j<\sigma^2_n\le \phi_{j-1}$, $2\le j \le d$, and notice that
\begin{align}
h(\sigma^2_n)-c(\sigma^2_n)&=\left[ (j-1)-\sum_{i=1}^d \frac{\phi_i}{\phi_1} \right]\sigma^2_n+\sum_{i=j}^d \phi_i
\nonumber
\\
&=\left[ \underbrace{(j-1)-\sum_{i=1}^{j-1} \frac{\phi_i}{\phi_1}}_{b\ge 0}-\sum_{i=j}^d\frac{\phi_i}{\phi_1} \right]\sigma^2_n
\nonumber
\\
&+\phi_1\sum_{i=j}^d \frac{\phi_i}{\phi_1}
\nonumber
\\
&=b\sigma^2_n+(\phi_1-\sigma^2_n)\sum_{i=j}^d \frac{\phi_i}{\phi_1} \ge 0.
\end{align}
Finally, it is not difficult to prove that the slope of the chord having endpoints $(0,a)$ and
$(y,\sum_{i=1}^d\phi_i)$ with $y>\phi_1$, is always lower than the slope of $c(\sigma^2_n)$.
Since the right-hand side of \eqref{eq:sigman-prop2} is strictly increasing and starts from $0$,
it is apparent that (\ref{eq:sigman-prop2})  admits one and only one solution 
for $\sigma^2_n>0$ that can be determined finding that value of $j$, say $\widehat{j}$, such that
the potential solution of the equation with 
\be
\phi_j < \sigma^2_n \leq \phi_{j-1}, \quad j=d+1, \ldots, 1,
\label{eq:interval-Bresler}
\ee
namely
\be
\overline{\sigma^2}(j)=
\begin{cases}
\displaystyle \frac{a}{N-d}, & j=d+1,
\\
\vspace{-1mm}
\\
\displaystyle \frac{1}{N-(j-1)}\left( a+ \sum_{i=j}^d \phi_i \right), & j=d, \ldots, 1,
\end{cases}
\label{eq:sigmaqj-Bresler}
\ee
actually belongs to the interval given by \eqref{eq:interval-Bresler}.
\end{proof}

\smallskip

The above results will be used in the next subsection to derive the
detection architecture.

\subsection{LRT-based Detector}
\label{Subsec:LRTdesign}
The design criterion consists in starting from the likelihood ratio test and
in replacing the unknown parameters with suitable estimates obtained by exploiting
all the available data. Specifically, the decision statistics is
\be
\frac{f_1(\bZ_P;\widehat{\sigma}^2_{n},\widehat{\bM}_{c}, \widehat{\bP}, \widehat{\bQ})}
{f_0(\bZ_P;\widehat{\sigma}^2_{n},\widehat{\bM}_{c},\widehat{\bQ})}\test\eta,
\label{eqn:LRT}
\ee
where $\widehat{\sigma}^2_{n}$ and $\widehat{\bQ}$ are the estimates of $\sigma^2_n$ and $\bQ$, respectively, 
based on $\bX$,
$\widehat{\bM}_{c}$ is the estimate of $\bM_c$ based on $\bZ_S$, and $\widehat{\bP}$ is the estimate
of $\bP$ over $\bZ_P$;
finally, $\eta$ is the detection threshold used to control
the probability of false alarm ($P_{fa}$).\footnote{Hereafter, we use symbol $\eta$ to denote
the generic detection threshold.}

The estimates under $H_1$ are computed according to the following steps
\begin{enumerate}
\item estimate $\sigma^2_n$ and $\bQ$ by maximizing the PDF of the additional data set $\bX$;
\item replace these estimates in the PDF of $\bZ_S$ and derive the maximum likelihood estimate of $\bM_c$;
\item consider the PDF of $\bZ_P$ and assume that $\sigma^2_n$, $\bQ$, 
and $\bM_c$ are known and equal to the previously obtained estimates; 
then, maximize the PDF of $\bZ_P$
with respect to $\bP$.
\end{enumerate}
Let us start from the estimation of $\sigma^2_n$ and $\bQ$. 
To this end, we suitably recast \eqref{eqn:PDF_X} to 
exploit Proposition \ref{prop:maxLogDetTr_C_PSD} and maximize
the log-likelihood function. Specifically, by applying the Woodbury identity and the equality
\be
\det(\bI+\bA\bB)=\det(\bI+\bB\bA)
\label{eqn:detAB_BA}
\ee
with $\bA$ and $\bB$ two matrices such that $\bA\bB$ and $\bB\bA$ 
are well-defined \cite{lutkepohl1997handbook}, 
we can write the logarithm of \eqref{eqn:PDF_X} as
\begin{align}
\cL_X(\sigma^2_n,\bQ)=&-K_JN\log\pi-K_JN\log\sigma^2_n
\nonumber
\\
&-K_J\log\det\left(\bI+\frac{\bQ\bGamma}{\sigma^2_n}\right)-\frac{K_J}{\sigma^2_n}\tr(\bS_X)
\nonumber
\\
&+\frac{K_J}{\sigma^2_n}\tr\left[\left(\sigma^2_n\bQ^{-1}+\bGamma\right)^{-1} \bS_{X,J}\right],
\label{eqn:LL_X_00}
\end{align} 
where for the moment we have assumed that $\bQ$ is nonsingular, 
$\bGamma=\bJ^\dag\bJ\in\C^{q\times q}$ is Hermitian positive definite, $\bS_X=\bX\bX^\dag/K_J$, and 
$\bS_{X,J}=\bJ^\dag\bS_X\bJ\in\C^{q\times q}$.
The Cholesky decomposition of $\bGamma$ allows us to recast \eqref{eqn:LL_X_00} as
\begin{align}
&\cL_X(\sigma^2_n,\bQ)=-K_JN\log\pi-K_JN\log\sigma^2_n
\nonumber
\\
&-K_J\log\det\left(\bI+\frac{\bQ\bL\bL^\dag}{\sigma^2_n}\right)-\frac{K_J}{\sigma^2_n}\tr(\bS_X)
\nonumber
\\
&+\frac{K_J}{\sigma^2_n}\tr\left[\left(\sigma^2_n\bQ^{-1}+\bL\bL^\dag\right)^{-1} \bS_{X,J}\right]
\\
&=-K_JN\log\pi-K_JN\log\sigma^2_n
\nonumber
\\
&-K_J\log\det\left(\bI+\frac{\bL^\dag\bQ\bL}{\sigma^2_n}\right)-\frac{K_J}{\sigma^2_n}\tr(\bS_X)
\nonumber
\\
&+\frac{K_J}{\sigma^2_n}\tr\left[\left(\bI+\sigma^2_n(\bL^\dag\bQ\bL)^{-1}\right)^{-1} 
\bL^{-1}\bS_{X,J}(\bL^\dag)^{-1}\right],
\label{eqn:LL_X_01}
\end{align}
where $\bL\in\C^{q\times q}$ is a lower triangular matrix with real and 
positive diagonal entries such that $\bGamma=\bL\bL^\dag$ and 
we have applied again \eqref{eqn:detAB_BA}. Let us make a change of variable by defining
$\bY=\bL^\dag\bQ\bL/\sigma^2_n$ and rewrite \eqref{eqn:LL_X_01} as follows
\begin{align}
\cL_X(\sigma^2_n,\bY)=&-K_JN\log\pi-K_JN\log\sigma^2_n
\nonumber
\\
&-K_J\log\det\left(\bI+\bY\right)-\frac{K_J}{\sigma^2_n}\tr(\bS_X)
\nonumber
\\
&+\frac{K_J}{\sigma^2_n}\tr\left[\left(\bI+\bY^{-1}\right)^{-1}\bS_{X,J,L}\right],
\label{eqn:LL_X_02}
\end{align}
where $\bS_{X,J,L}=\bL^{-1}\bS_{X,J}(\bL^\dag)^{-1}$. Now, it is not difficult to show that
\begin{align}
(\bI+\bY^{-1})^{-1}&=[\bY^{-1}(\bI+\bY)]^{-1}=(\bI+\bY)^{-1}\bY
\nonumber
\\
&=\bI-(\bI+\bY)^{-1}.
\end{align}
It follows that
\begin{align}
\cL_X(\sigma^2_n,\bY)=&-K_JN\log\pi-K_JN\log\sigma^2_n
\nonumber
\\
&-K_J\log\det\left(\bI+\bY\right)-\frac{K_J}{\sigma^2_n}\tr(\bS_X-\bS_{X,J,L})
\nonumber
\\
&-K_J\tr\left[\left(\bI+\bY\right)^{-1}\frac{\bS_{X,J,L}}{\sigma^2_n}\right].
\label{eqn:LL_X_03}
\end{align}
It is clear that the above expression no longer requires that $\bQ$ is nonsingular. Equality between \eqref{eqn:LL_X_03} and \eqref{eqn:PDF_X} for positive semidefinite $\bQ$ follows immediately. As a matter of fact,
for any positive semidefinite $\bQ$  and $\forall \epsilon >0$, 
$\bQ+ \epsilon \bI$ is positive definite  and expressions \eqref{eqn:LL_X_03} and \eqref{eqn:PDF_X}  are equal. Thus, they are equal also in the limit as $\epsilon \rightarrow 0$,
since they are continuous with respect to $\epsilon$.

To maximize $\cL_X(\sigma^2_n,\bY)$ with respect to $\bY$, 
we apply Proposition \ref{prop:maxLogDetTr_C_PSD} with
$\bC=\bY$ and $\bR=\bS_{X,J,L}/\sigma^2_n$. It follows that
$\bU$ contains the eigenvectors of $\bS_{X,J,L}$, 
$\lambda_i=\beta_i/\sigma^2_n$, $i=1,\ldots,q$,
with $\beta_1\ge\ldots\ge\beta_{q}\ge 0$ the eigenvalues of $\bS_{X,J,L}$,
and
\be
\widehat{\bY}=\bU \diag\!\!\left(\!\max\left\{\!\frac{\beta_1}{\sigma^2_n}-1,0\!\right\},\ldots,
\max\left\{\!\frac{\beta_{q}}{\sigma^2_n}-1,0\!\right\}\!\right)\!\bU^\dag.
\ee
Replacing $\bY$ with $\widehat{\bY}$ in \eqref{eqn:LL_X_03} leads to
\begin{align}
\cL_X(\sigma^2_n,\widehat{\bY})=&-K_JN\log\pi-K_JN\log\sigma^2_n
\nonumber
\\
&-K_J\sum_{i=1}^{q}\log\left[\max\left(\frac{\beta_i}{\sigma^2_n},1\right)\right]
\nonumber
\\
&-\frac{K_J}{\sigma^2_n}\tr(\bP_J^\perp\bS_X)
\nonumber
\\
&-K_J\sum_{i=1}^{q}\frac{\beta_i/\sigma^2_n}{\max(\beta_i/\sigma^2_n,1)},
\label{eqn:LL_X_04}
\end{align}
where we have used the following equality
\begin{align}
\tr(\bS_X-\bS_{X,J,L})&=\tr(\bS_X)-\tr(\bS_{X,J,L})
\nonumber
\\
&=\tr(\bS_X)-\tr(\bJ(\bL\bL^\dag)^{-1}\bJ^\dag\bS_{X})
\nonumber
\\
&=\tr(\bS_X-\bP_J\bS_X)
\nonumber
\\
&=\tr(\bP^\perp_J\bS_X)
\end{align}
with $\bP_J=\bJ(\bJ^\dag\bJ)^{-1}\bJ^\dag$ and $\bP_J^\perp=\bI-\bP_J$.

The final estimation of this step concerns $\sigma^2_n$. This estimate
is easily performed resorting to Proposition \ref{stima_sigmaquadro} with $q$ in place of $d$ and the $\beta_i$s in place of the $\phi_i$s. 
Denoting by $\widehat{\sigma}^2_{n}$ the estimate of $\sigma^2_n$, the final estimate of $\bQ$ is given by
\begin{multline}
\widehat{\bQ}=\widehat{\sigma}^2_{n} (\bL^\dag)^{-1}
\bU
\\
\times 
\diag\left(\frac{\beta_1}{\widehat{\sigma}^2_{n}}-1,\ldots,\frac{\beta_{\widehat{j}-1}}
{\widehat{\sigma}^2_{n}}-1,0,\ldots,0\right)
\bU^\dag\bL^{-1}
\end{multline}
if $\widehat{j}>1$
and
$$
\widehat{\bQ}= \bzero
$$
if instead $\widehat{j}=1$ (see the definition of $\widehat{j}$ before \eqref{eq:interval-Bresler}).
The first step is thus complete.

In the second step, we use the PDF of $\bZ_S$ to estimate $\bM_c$
under the assumption that $\sigma^2_n=\widehat{\sigma}^2_{n}$ and
$\bQ=\widehat{\bQ}$. Using \eqref{eqn:PDF_Zs_H1_3}, the logarithm of the 
PDF of $\bZ_S$ can be recast as
\begin{align}
\cL_S(\bM_c)=&-K_SN\log\pi-K_S\log\det(\bT+\bM_c)
\nonumber
\\
&-K_S\tr[(\bT+\bM_c)^{-1}\bS_S]
\nonumber
\\
=&-K_SN\log\pi-K_S\log\det(\bT)
\nonumber
\\
&-K_S\log\det(\bI+\bM_{c,T})
\nonumber
\\
&-K_S\tr[(\bI+\bM_{c,T})^{-1}\bS_{S,T}],
\end{align}
where $\bT=\widehat{\sigma}^2_{n}\bI+\bJ\widehat{\bQ}\bJ^\dag$, $\bS_S=\bZ_S\bZ_S^\dag/K_S$,
$\bS_{S,T}=\bT^{-1/2}\bS_{S}\bT^{-1/2}$, and $\bM_{c,T}=\bT^{-1/2}\bM_c\bT^{-1/2}$.
Exploiting Proposition \ref{prop:maxLogDetTr_C_PSD} with $\bC=\bM_{c,T}$
and $\bR=\bS_{S,T}$, we can maximize $\cL_S(\bM_c)$ with respect to
$\bM_{c,T}$ and, as a byproduct, to $\bM_c$. In this case, $\bU$ contains 
the eigenvectors of $\bS_{S,T}$ whose eigenvalues
are denoted by $\omega_1\ge \ldots\ge \omega_N\ge 0$.
Thus, the maximum likelihood 
estimate of $\bM_c$ over $\bZ_S$ is given by
\begin{multline}
\widehat{\bM}_{c}=\bT^{1/2}\bU
\diag(
\max\{\omega_1-1,0\},
\\
\ldots,\max\{\omega_N-1,0\})
\bU^\dag
\bT^{1/2}.
\end{multline}

Finally, we accomplish the last step by suitably recasting the PDF of $\bZ_P$ under
the assumption that $\bM_c=\widehat{\bM}_{c}$, $\sigma^2_n=\widehat{\sigma}^2_{n}$, and
$\bQ=\widehat{\bQ}$. Specifically, let us define  
$\bG=\widehat{\sigma}^2_{n}\bI+\widehat{\bM}_{c}+\bJ\widehat{\bQ}\bJ^\dag$ and write
the logarithm of the PDF of $\bZ_P$ as follows
\begin{align}
\cL_P(\bP)=&-K_PN\log\pi-K_P\log\det(\bG)
\nonumber
\\
&-K_P\log\det(\bI+\bH_{G}\bP\bH_{G}^\dag)
\nonumber
\\
&-K_P\tr[(\bI+\bH_{G}\bP\bH_{G}^\dag)^{-1}\bS_{G}],
\end{align}
where $\bH_{G}=\bG^{-1/2}\bH$ and $\bS_{G}=\bG^{-1/2}\bS_P\bG^{-1/2}$ with $\bS_P=\bZ_P\bZ_P^\dag/K_P$.
Following the same line of reasoning as for the estimation of $\bQ$ by means of $\bX$,
$\cL_P(\bP)$ can be further recast as
\begin{align}
\cL_P(\bP)=&-K_PN\log\pi-K_P\log\det(\bG)
\nonumber
\\
&-K_P\log\det(\bI+\bL_{G}^\dag\bP\bL_{G})
\nonumber
\\
&-K_P\tr(\bS_{G}-\bS_{G,H,L})
\nonumber
\\
&-K_P\tr[(\bI+\bL_{G}^\dag\bP\bL_{G})^{-1}\bS_{G,H,L}],
\end{align}
where $\bL_G\in\C^{r\times r}$ is a lower triangular matrix obtained from the
Cholesky decomposition of $\bH_{G}^\dag\bH_{G}\in\C^{r\times r}$ and 
$\bS_{G,H,L}=\bL_G^{-1}\bH_G^\dag\bS_{G}\bH_G(\bL_G^\dag)^{-1}\in\C^{r\times r}$.
By Proposition \ref{prop:maxLogDetTr_C_PSD}, the maximization of $\cL_P(\bP)$ with respect to $\bP$ leads to
\begin{multline}
\widehat{\bP}=(\bL_G^\dag)^{-1}
\bU
\\
\times 
\diag\left(\max\{\phi_1-1,0\},\ldots,\max\{\phi_r-1,0\}\right)
\bU^\dag\bL_G^{-1},
\end{multline}
where $\phi_1\ge\ldots\ge\phi_r\ge 0$ are the eigenvalues of $\bS_{G,H,L}$
and the columns of $\bU$ are the respective eigenvectors.

As for the estimates under $H_0$, we proceed by repeating the first two steps 
of the procedure under $H_1$, namely
\begin{enumerate}
\item compute the maximum likelihood estimates of $\sigma^2_n$ and $\bQ$ by using the additional
data set $\bX$;
\item replace the estimates at the previous step in the PDF of $\bZ_S$ and
derive the maximum likelihood estimate of $\bM_c$.
\end{enumerate}
Then, we  replace all the above estimates in the PDF of $\bZ_P$ under $H_0$ to obtain
\begin{align}
\cL_0(\widehat{\sigma}^2_{n},&\widehat{\bM}_{c},\widehat{\bQ})=-K_PN\log\pi
\nonumber
\\
&-K_P\log\det\left(\widehat{\sigma}^2_{n}\bI+\widehat{\bM}_{c}+\bJ\widehat{\bQ}\bJ^\dag\right)
\nonumber
\\
&-K_P\tr\left[\left(\widehat{\sigma}^2_{n}\bI+\widehat{\bM}_{c}+\bJ\widehat{\bQ}\bJ^\dag\right)^{-1}\bS_P\right].
\end{align}

In what follows, we refer to \eqref{eqn:LRT} as second-order jamming-affected LRT (SOJ-LRT).

\section{An Identifiability Issue}
\label{Sec:Identifiability}
In this section, we show that estimating $\bQ$ and $\bM$ by using $\bZ_S$ only
is affected by an identifiability issue. To this end, let us consider the terms
of the log-likelihood function depending on $\bQ$ and $\bM=\sigma^2_n\bI+\bM_c$, namely
\be
-K_S\log\det(\bM+\bJ\bQ\bJ^\dag)-K_S\tr[
(\bM+\bJ\bQ\bJ^\dag)^{-1}\bS_S]
\ee
and recast them by using the QR-decomposition of $\bJ=\bU_j\bR$, where $\bU_j\in\C^{N\times q}$
is a slice of a unitary matrix (i.e., $\bU_j^\dag\bU_j=\bI$) and $\bR\in\C^{q\times q}$
is an upper triangular (that is invertible since $\bJ$ is full-column rank), as follows
\begin{multline}
-K_S\log\det\left(\bM+\bU_j\bQ_R\bU_j^\dag\right)
\\
-K_S\tr\left[\left(\bM+\bU_j\bQ_R\bU_j^\dag\right)^{-1}\bS_S\right],
\end{multline}
where $\bQ_R=\bR\bQ\bR^\dag$. Now, let $\bV\in\C^{N\times N}$ be a unitary matrix such that
$\bV\bU_j=\bE=[\bI \ \bzero]^T$, where the identity matrix has size $q\times q$ and recast
the objective function in terms of $\bE$ as follows
\begin{multline}
-K_S\log\det\left(\bM_V+\bE\bQ_R\bE^\dag\right)
\\
-K_S\tr\left[\left(\bM_V+\bE\bQ_R\bE^\dag\right)^{-1}\bS_{SV}\right],
\label{eqn:identifiability_obj_1}
\end{multline}
where $\bM_V=\bV\bM\bV^\dag$ and $\bS_{SV}=\bV\bS_S\bV^\dag$.
Finally, partitioning $\bM_V$ and $\bS_{SV}$ as follows
\be
\bM_V=
\begin{bmatrix}
\bM_{V,11} & \bM_{V,12}
\\
\bM_{V,21} & \bM_{V,22}
\end{bmatrix}, \ 
\bS_{SV}=\begin{bmatrix}
\bS_{SV,11} & \bS_{SV,12}
\\
\bS_{SV,21} & \bS_{SV,22}
\end{bmatrix},
\ee
where $\bM_{V,11}\in\C^{q\times q}$ and $\bS_{SV,11}\in\C^{q\times q}$,
it turns out that
\[
\bM_V+\bE\bQ_R\bE^\dag=
\begin{bmatrix}
\bM_{V,11}+ \bQ_R & \bM_{V,12}
\\
\bM_{V,21} & \bM_{V,22}
\end{bmatrix}.
\]
It is thus apparent that
different (actually infinite) choices of $\bM_{V,11}$ and $\bQ_R$ return the same value of
eq. \eqref{eqn:identifiability_obj_1}.

\section{Performance Assessment}
\label{Sec:PerformanceAssessment}
In this section, we assess the performance of the proposed SOJ-LRT 
against a suitable counterpart by resorting to standard Monte Carlo (MC)
counting techniques. 
After describing the simulation setup and the benchmark detector, we investigate the 
behavior of the SOJ-LRT false-alarm probability ($P_{fa}$) and show 
that a fixed threshold does not ensure its regulation when 
the clutter and jammer powers vary. For this reason, we modify the detection architecture
by introducing and validating an adaptive threshold-selection strategy based 
on estimated disturbance parameters. Finally, the detection performance 
of SOJ-LRT is assessed for different amounts of secondary and additional training 
data in comparison with that of the benchmark detector.
Specifically, the detection performance is evaluated in terms of probability of detection
$(P_d)$ for a given $P_{fa}=10^{-4}$ (nominal $P_{fa}$); the $P_d$ is estimated
over $1000$ independent MC trials; as for the detection thresholds,
they are set by using $100/P_{fa}$ MC independent trials.
In all the illustrative examples, we assume $N=16$ spatial channels,
$r=3$, $q=4$, $K_P=16$, $K_S\in\{12,16,24\}$, and $K_J\in\{12,16,24\}$.

The covariance clutter component is generated by assuming that clutter returns
impinge the system from the angular sector $[-50^\circ,-30^\circ]$. To this end, we
compute the prolate matrix $\bR_{\beta_C} \in \C^{N \times N}$
whose $(m,n)$th entry has the following expression \cite{VanTrees4,4490112}
\be
\bR_{\beta_C}(m,n) = 2 e^{j(m-n)\pi \gamma_{1,C}} \gamma_{2,C} \pi \sinc((n-m)\gamma_{2,C})
\ee
where $\gamma_{1,C}=\sin(-40\pi/180)\cos(\theta_C)$ and $\gamma_{2,C}=\cos(-40\pi/180)\sin(\theta_C)$
with $\theta_C=10\pi/180$.
Then, we compute the singular value decomposition of $\bR_{\beta_C}$, select the
left unitary matrix $\bU_C\in\C^{N\times N}$, and define $\bM_c$ as
\be
\bM_c=\frac{\sigma^2_c}{\tr[\bU_C\bU_C^\dag]}\bU_C\bU_C^\dag
\ee
where $\sigma^2_c>0$ is the clutter power set according
to a clutter-to-noise ratio (CNR) of $30$ dB and defined as $\sigma^2_c/\sigma^2_n$
with $\sigma^2_n=1$.

In the simulated scenario 
the signal (jamming) component in the $i$th primary data vector $\bz_{p,i}$, $i=1,\ldots,K_P$,
is given by $\alpha_{T,i}\bv(\phi_{T,i})$ ($\alpha_{J,i}\bv(\phi_{J,i})$), with
\[
\bv(\theta) = \frac{1}{\sqrt{N}} \left[ 1 \ e^{j\theta} \ \cdots \ e^{j(N-1)\theta } \right]^T.
\] 
The electrical angles $\phi_{T,i}$ ($\phi_{J,i}$) are independent random variables uniformly 
distributed over $\Phi_T=\left[ -\pi \beta_T, \pi \beta_T \right]$ 
($\Phi_J=\left[\pi\beta_{J,l}, \pi\beta_{J,u} \right]$), 
where $\beta_T=\sin \theta_T$ ($\beta_{J,l}=\sin(20\pi/180-\theta_J)$, 
$\beta_{J,u}=\sin(20\pi/180+\theta_J)$) 
and $\theta_T=2\pi/180$ ($\theta_J=2\pi/180$) radians. 
We form the signal (jamming) subspace 
by computing the prolate matrix $\bR_{\beta_T} \in \C^{N \times N}$ ($\bR_{\beta_J} \in \C^{N \times N}$), 
whose $(m,n)$th entry has the following expression \cite{VanTrees4,4490112}
\be
\bR_{\beta_T}(m,n) = 2 \beta_T \pi \sinc((n-m)\beta_T)
\ee
($\bR_{\beta_J}(m,n) = 2 e^{j(m-n)\pi \gamma_{1,J}} \gamma_{2,J} \pi \sinc((n-m)\gamma_{2,J})$, where 
$\gamma_{1,J}=\sin(20\pi/180)\cos(\theta_J)$ and $\gamma_{2,J}=\cos(20\pi/180)\sin(\theta_J)$).
Then, the signal (jamming) subspace is chosen as that spanned by the 
first $r$ ($q$) eigenvectors of $\bR_{\beta_T}$ ($\bR_{\beta_J}$) corresponding to
the most significant eigenvalues.

The jamming components in $\bZ_S$ (secondary data) and $\bX$ (additional data set) 
are generated using the same procedure as for the primary data.

As for the target and jamming amplitudes, they are generated as independent circularly symmetric
complex Gaussian random variables with zero mean and variance $\sigma^2_T>0$ for the target and $\sigma^2_J>0$
for the jammer. The signal-to-noise ratio (SNR) for a single range bin, defined as $\sigma^2_T/\sigma^2_n$, 
is used to set $\sigma^2_T>0$, while the jammer-to-noise ratio (JNR) for a single range bin, 
defined as $\sigma^2_J/\sigma^2_n$, allows us to set $\sigma^2_J$.

Finally, we provide the expression of the competitor. To this end, denoting
by $\bW_S=\frac{1}{\sqrt{K_S}}[\bS_S^{-1/2}\bH, \ \bS_S^{-1/2}\bJ]$, it is
given by \cite{SD}
\be
\label{eqn:Competitor}
\frac{\det\left[\bI+\frac{1}{K_S}(\bS_S^{-1/2}\bZ_P)^\dag(\bI-\bP_{\bf J_S})(\bZ_P\bS_S^{-1/2})\right]}
{\det\left[
\bI+\frac{1}{K_S}(\bS_S^{-1/2}\bZ_P)^\dag
(\bI-\bP_{\bf W_S})
(\bS_S^{-1/2}\bZ_P)
\right]}\test\eta,
\ee
where $\bP_{\bf W_S}=\bW_S\left(\bW_S^\dag\bW_S\right)^{-1}\bW_S^\dag$ and
$\bP_{\bf J_S}=\bS_S^{-1/2}\bJ\left(\bJ^\dag\bS_S^{-1}\bJ\right)^{-1}\bJ^\dag\bS_S^{-1/2}$.
Notice that the test statistic of the above detector can be implemented if and only if $\bS_S$
is invertible with probability $1$, namely when $K_S\ge N$.

\bigskip

\begin{figure*}[t]
    \centering
    \begin{subfigure}[t]{0.32\textwidth}
        \centering
        \includegraphics[width=\linewidth]{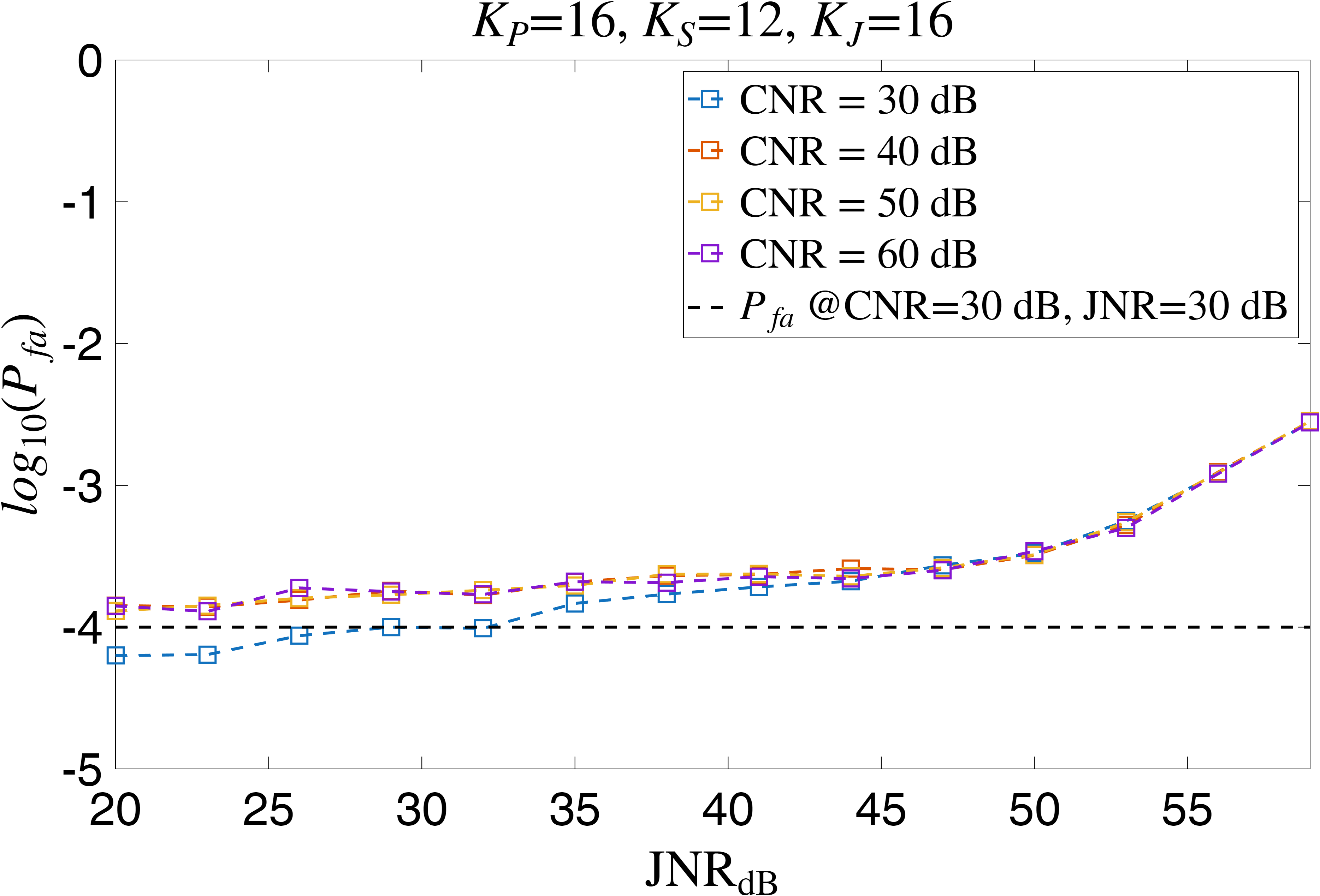}
        \caption{}
        \label{fig:a}
    \end{subfigure}
    \hfill
    \begin{subfigure}[t]{0.32\textwidth}
        \centering
        \includegraphics[width=\linewidth]
        {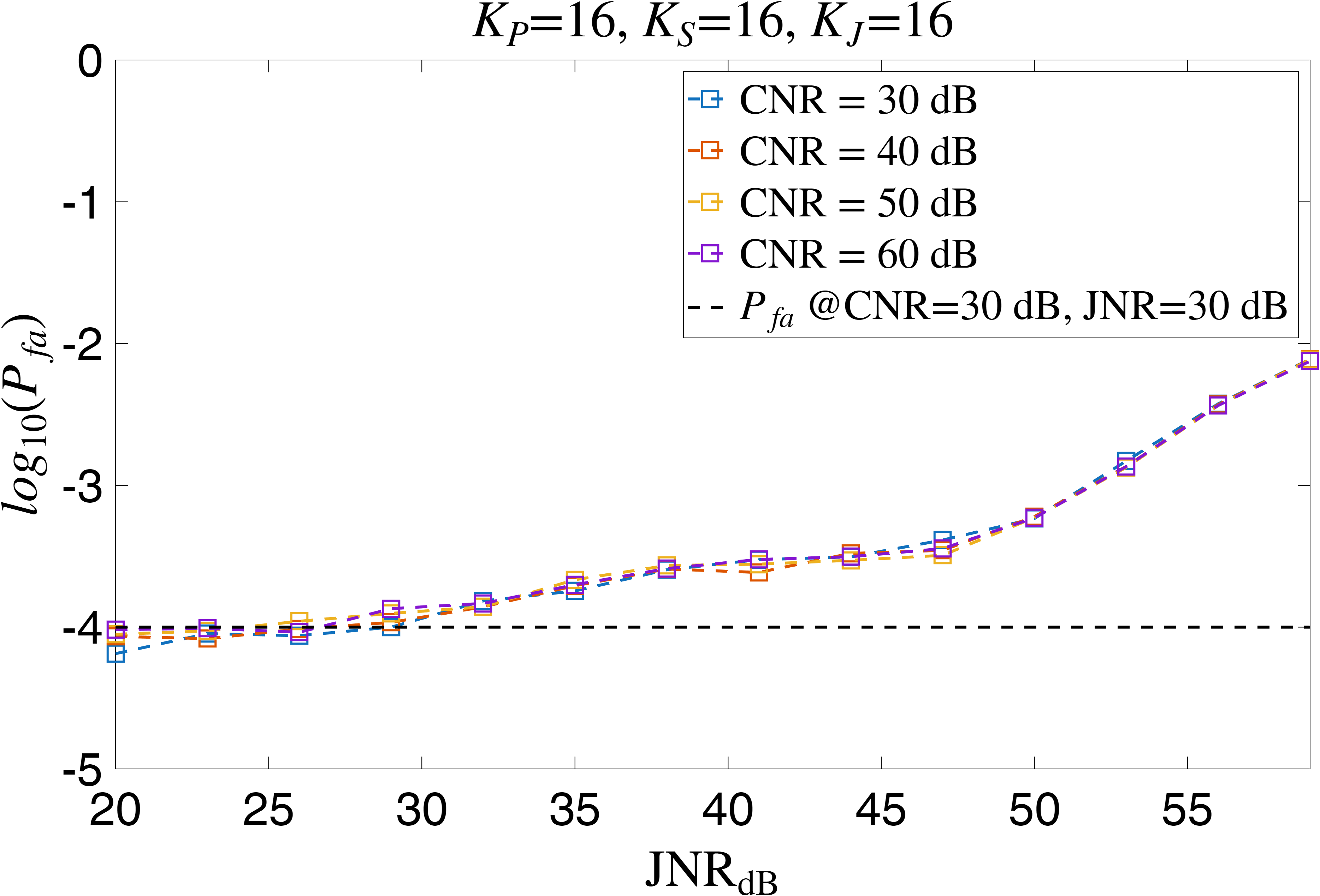}
        \caption{}
        \label{fig:b}
    \end{subfigure}
    \hfill
    \begin{subfigure}[t]{0.32\textwidth}
        \centering
        \includegraphics[width=\linewidth]
        {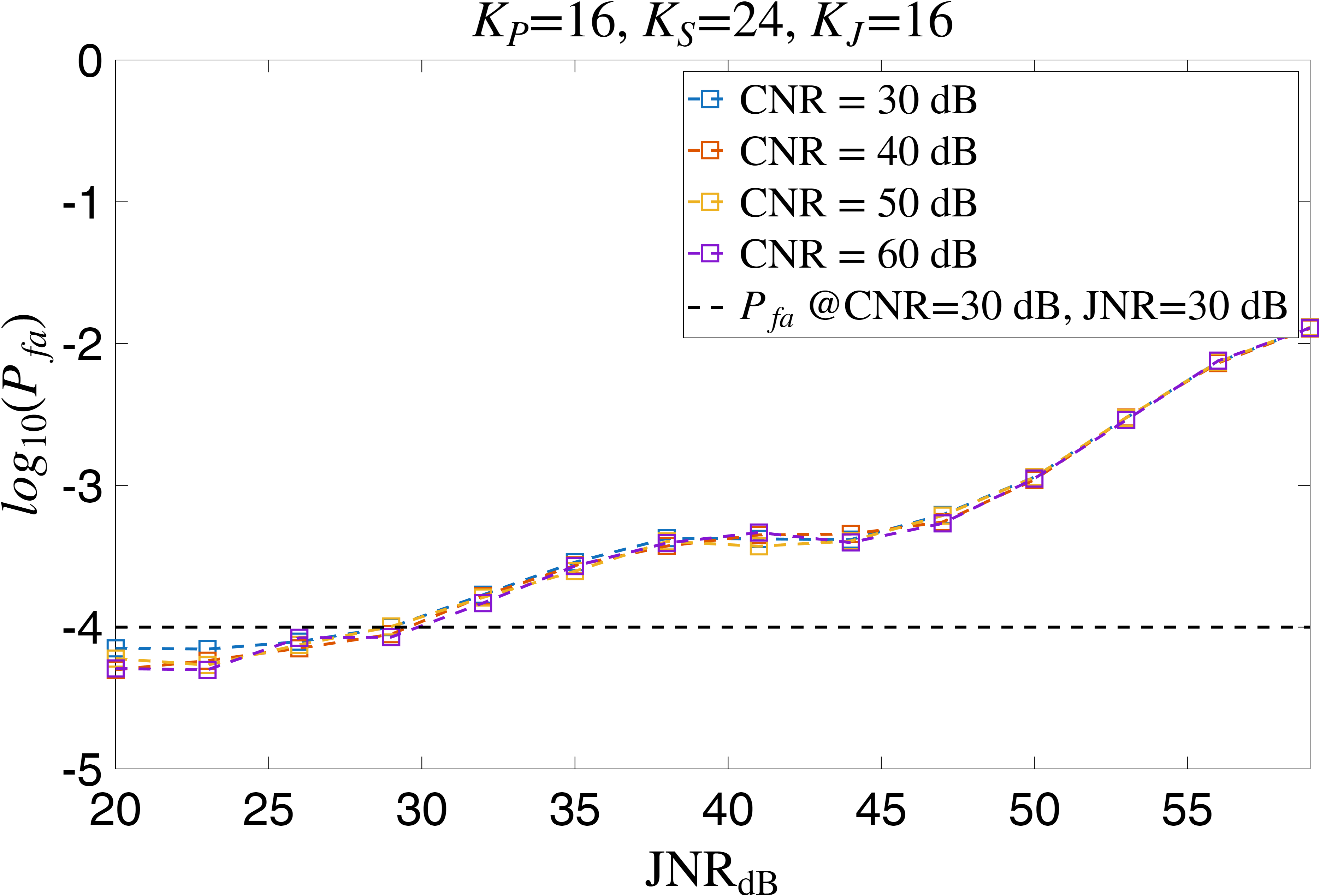}
        \caption{}
        \label{fig:c}
    \end{subfigure}

    \medskip

    \begin{subfigure}[t]{0.32\textwidth}
        \centering
        \includegraphics[width=\linewidth]
        {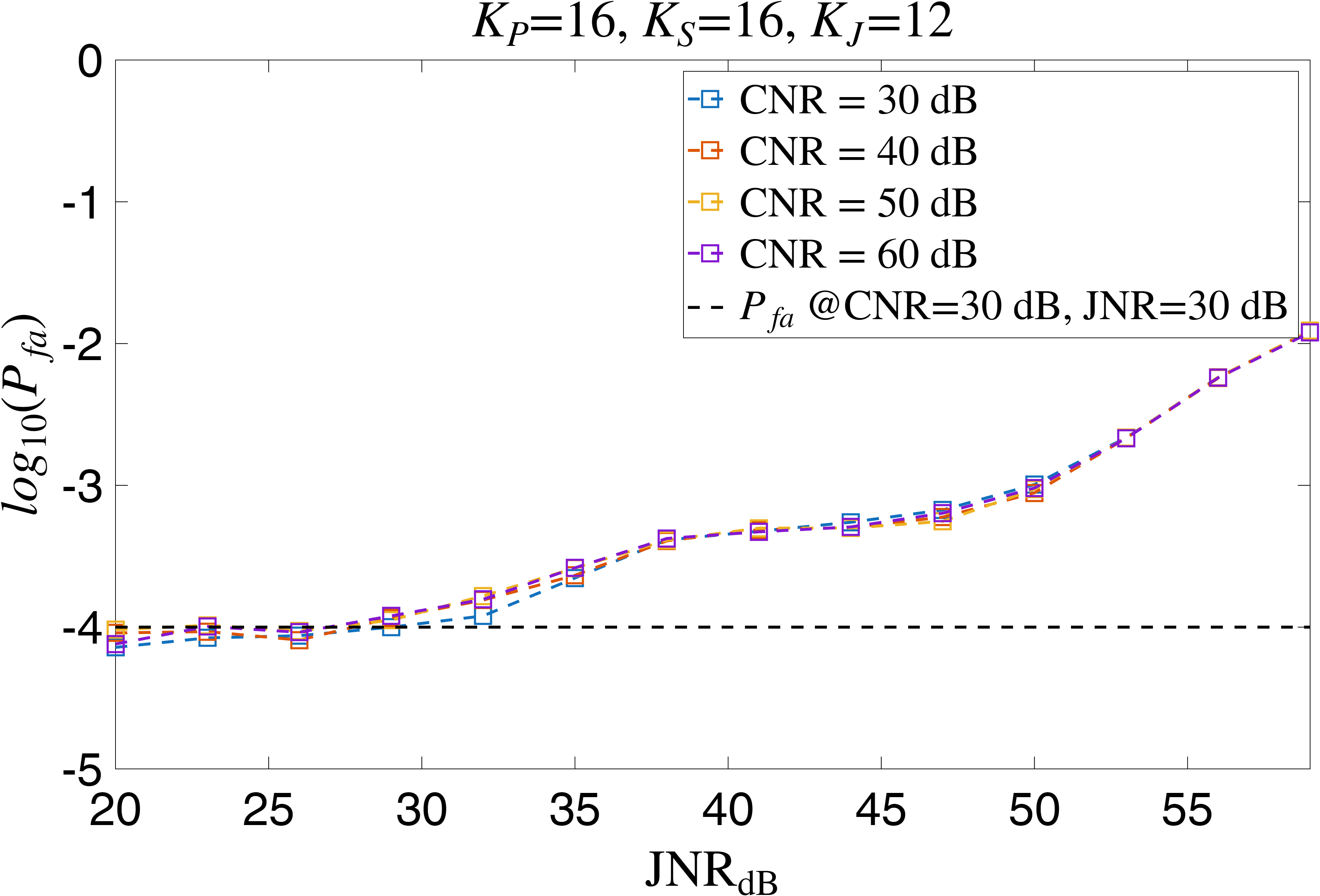}
        \caption{}
        \label{fig:d}
    \end{subfigure}
    \hspace{0.03\textwidth}
    \begin{subfigure}[t]{0.32\textwidth}
        \centering
        \includegraphics[width=\linewidth]
        {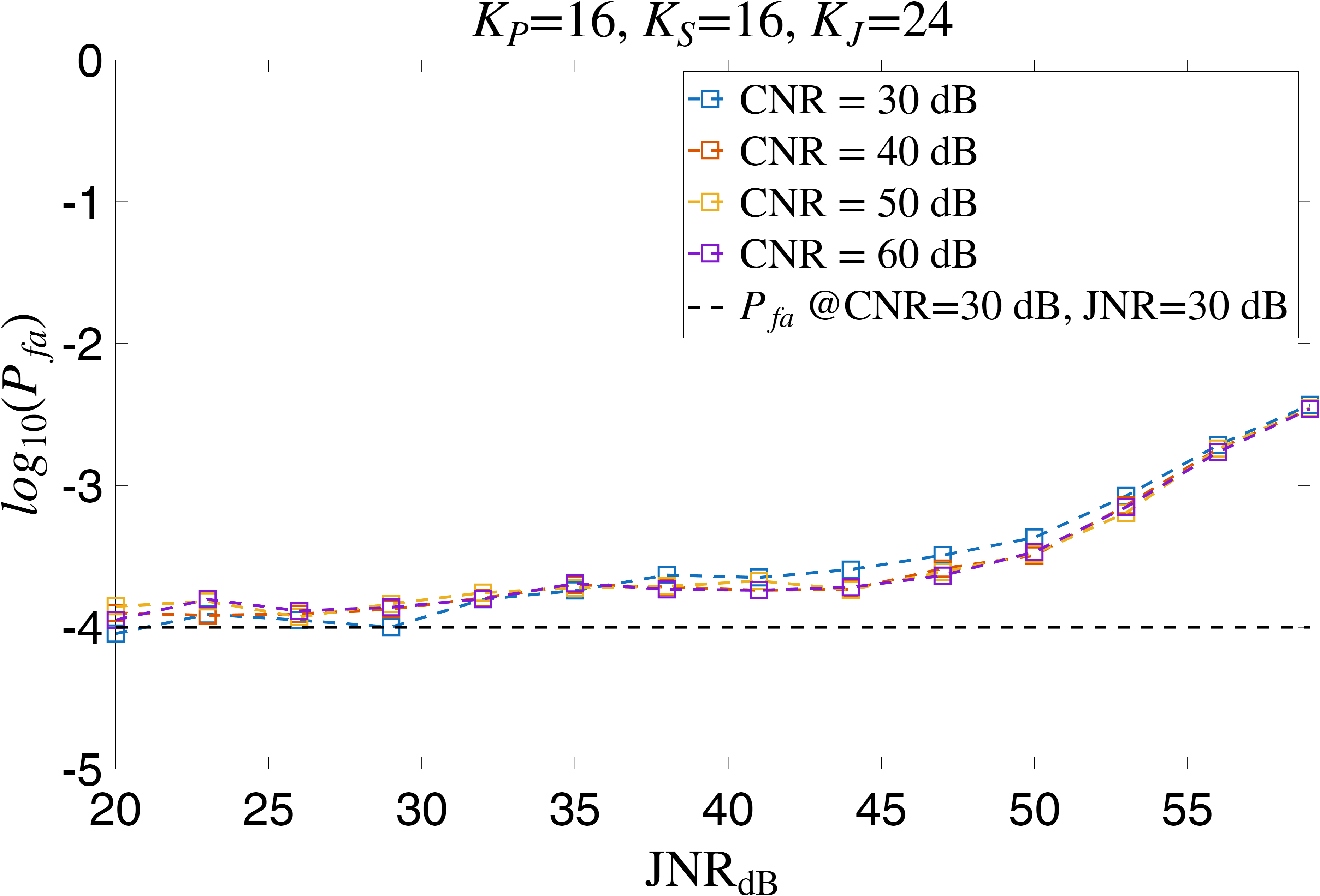}
        \caption{}
        \label{fig:e}
    \end{subfigure}

    \caption{Actual $P_{fa}$ as a function of JNR for $\mbox{CNR}\in\{30,40,50,60\}$ dB
    and five data size configurations: $(K_P,K_S,K_J)=(16,12,16)$ (a), 
    $(16,16,16)$ (b), $(16,24,16)$ (c), $(16,16,12)$ (d), and $(16,16,24)$ (e). 
    For each configuration, the threshold is set to provide the nominal 
    $P_{fa}$ ($10^{-4}$) at the reference operating points $\mbox{CNR}=\mbox{JNR}=30$ dB. 
    The horizontal dashed line denotes the nominal $P_{fa}$.}
    \label{fig:PFA_analysis}
\end{figure*}

\subsection{False-Alarm Analysis}
We first investigate the sensitivity of the SOJ-LRT $P_{fa}$ 
to the clutter and jammer power variations when a given detection threshold is employed. 
Since the distribution of the test statistic also depends on the number of primary ($K_P$), 
secondary ($K_S$), and additional training data ($K_J$), the analysis is 
conducted separately for different triplets $(K_P,K_S,K_J)$.

For each considered triplet, the detection threshold $\eta$ is set
so as to obtain the nominal $P_{fa}$ ($10^{-4}$) with the parameters 
CNR and JNR both set to the reference value $30$ dB. The same threshold 
is used to estimate the actual $P_{fa}$ when data are 
generated (under the $H_0$ hypothesis) by varying the
values of CNR and JNR. Specifically, the CNR takes on values in $\{30,40,50,60\}$ dB 
whereas the JNR values range from $20$ to $59$ dB in steps of $3$ dB. 

Five configurations for $(K_P,K_S,K_J)$ are considered: $(16,12,16)$, 
$(16,16,16)$, $(16,24,16)$, $(16,16,12)$, $(16,16,24)$. Specifically, $K_P$ is equal 
to $16$, while $K_S$ and $K_J$ take on values in $\{12,16,24\}$ in order to assess the 
dependence on the number of secondary ($\bZ_S$) and additional ($\bX$) data.

Figure \ref{fig:PFA_analysis} shows the obtained results, where for each 
triplet the $P_{fa}$ is plotted as a function of JNR and for the different values of CNR.
The figure shows that the SOJ-LRT is not inherently CFAR with respect to both JNR and CNR. 
In fact, we observe a pronounced dependence of $P_{fa}$ on the JNR, with deviations 
exceeding one order of magnitude at high jammer powers. Conversely, 
the CNR dependence is generally weak and becomes 
noticeable mainly in the sample-starved case ($K_S=12<N$) and 
for low values of CNR, when clutter covariance estimation is less accurate. 

To overcome this limitation, we introduce an adaptive strategy for 
detection threshold setting so as to keep $P_{fa}$ approximately 
constant as the JNR varies. Specifically, for each triplet $(K_P,K_S,K_J)$
and for the CNR reference value of $40$ dB, a lookup 
table (LUT) is constructed by varying the JNR from $20$ to $59$ dB in $3$ dB increments. 
The entries of this LUT are the thresholds (estimated over $100/P_{fa}$ independent
MC trials) corresponding to the nominal $P_{fa}=10^{-4}$ and to each value of JNR.
Then, during the detection stage, the adaptive 
threshold $\widehat{\eta}$ is obtained by linearly interpolating 
the LUT at the estimated JNR ($\widehat{\mathrm{JNR}}$).
The latter is obtained by exploiting intermediate 
estimates already provided by the procedure described 
in Subsection \ref{Subsec:LRTdesign} 
from the additional data $\bX$. Such estimates are used 
to compute $\widehat{\mathrm{JNR}}$ 
as\footnote{Notice that in the $i$th range bin, the jammer 
component is $\alpha_{J,i}\bv(\phi_{J,i})$ and approximately belongs to the subspace $\bJ$ (that is a slice
of unitary matrix).
It follows that
\be
\alpha_{J,i}\bv(\phi_{J,i})\approx\bJ\bq_i,
\ee
where $\bq_i$ is a random vector with zero mean and covariance matrix $\bQ$.
The covariance matrices of both vectors can be written as
\be
\sigma^2_J E[\bv(\phi_{J,i})\bv(\phi_{J,i})^\dag] \approx \bJ\bQ\bJ^\dag,
\ee
where $E[\cdot]$ denotes the statistical expectation. Computing the trace of both sides of
the above equation leads to
\be
\sigma^2_J \approx \tr[\bQ].
\ee
}
\be
\label{eqn:JNR_estimation}
\widehat{\mathrm{JNR}}
=
\frac{\operatorname{Tr}(\widehat{\boldsymbol{Q}})}
     {\widehat{\sigma}_n^2}.
\ee
The effectiveness of the adaptive threshold-selection strategy 
is assessed for $\mbox{CNR}\in\{35,45,55\}$ dB and also including JNR values not 
belonging to the LUT, that is $\mbox{JNR}\in\{22.5,33.5,37.5,45.5,52.5,56.5\}$ dB. 
The resulting false alarm probabilities 
are reported in Figure \ref{fig:PFA_analysis_adapt_th}.
The curves of $P_{fa}$ refer to all five triplets $(K_P,K_S,K_J)$ considered 
in Figure \ref{fig:PFA_analysis}. Moreover, Figure \ref{fig:PFA_adapt_CNR35} 
reports the $P_{fa}$ values obtained at $\mbox{CNR}=35$ dB for two sets 
of JNR values: the LUT grid points, ranging from $20$ to $59$ dB in $3$ dB increments, 
and the additional off-grid points. Figures \ref{fig:PFA_adapt_CNR45} and 
\ref{fig:PFA_adapt_CNR55} show analogous results except for $\mbox{CNR}=45$ dB
and $\mbox{CNR}=55$ dB, respectively. The figures
clearly highlight that the adaptive threshold-selection 
strategy approximately maintains $P_{fa}$ at its nominal value 
over the considered CNR and JNR ranges and for all the 
examined triplets $(K_P,K_S,K_J)$. Slightly larger deviations from 
the nominal $P_{fa}$ are observed for the sample-starved configuration $K_S=12<N$. 
This behavior is consistent with the lower accuracy of the clutter covariance 
estimate and with the residual CNR dependence that is not 
explicitly accounted for by the one-dimensional LUT. 
Nevertheless, the deviations remain substantially smaller than those observed when a fixed threshold is employed. 

\bigskip

\begin{figure*}[t]
    \centering
    \begin{subfigure}[t]{0.32\textwidth}
        \centering
        \includegraphics[width=\linewidth]
        {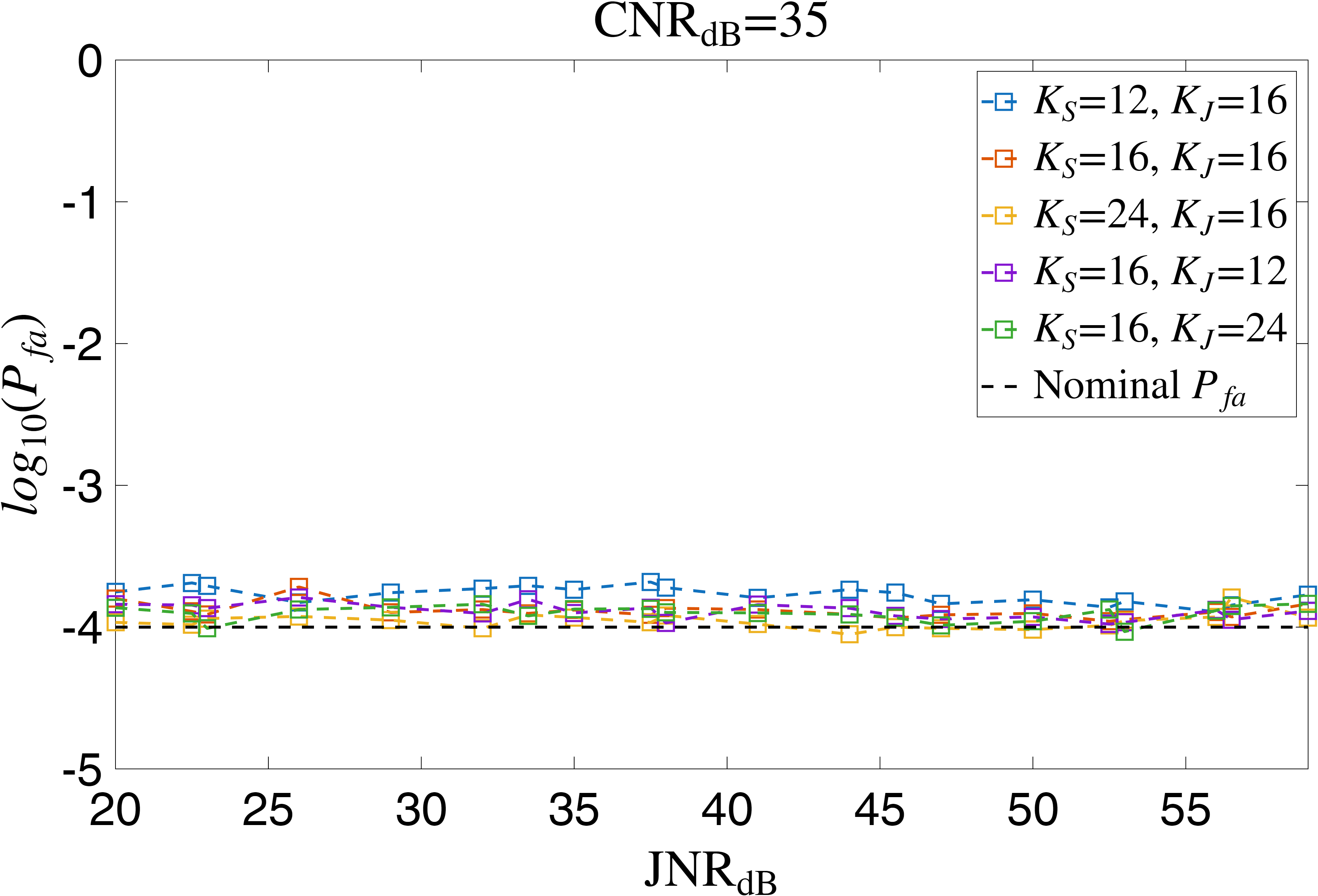}
        \caption{}
        \label{fig:PFA_adapt_CNR35}
    \end{subfigure}
    \hfill
    \begin{subfigure}[t]{0.32\textwidth}
        \centering
        \includegraphics[width=\linewidth]
        {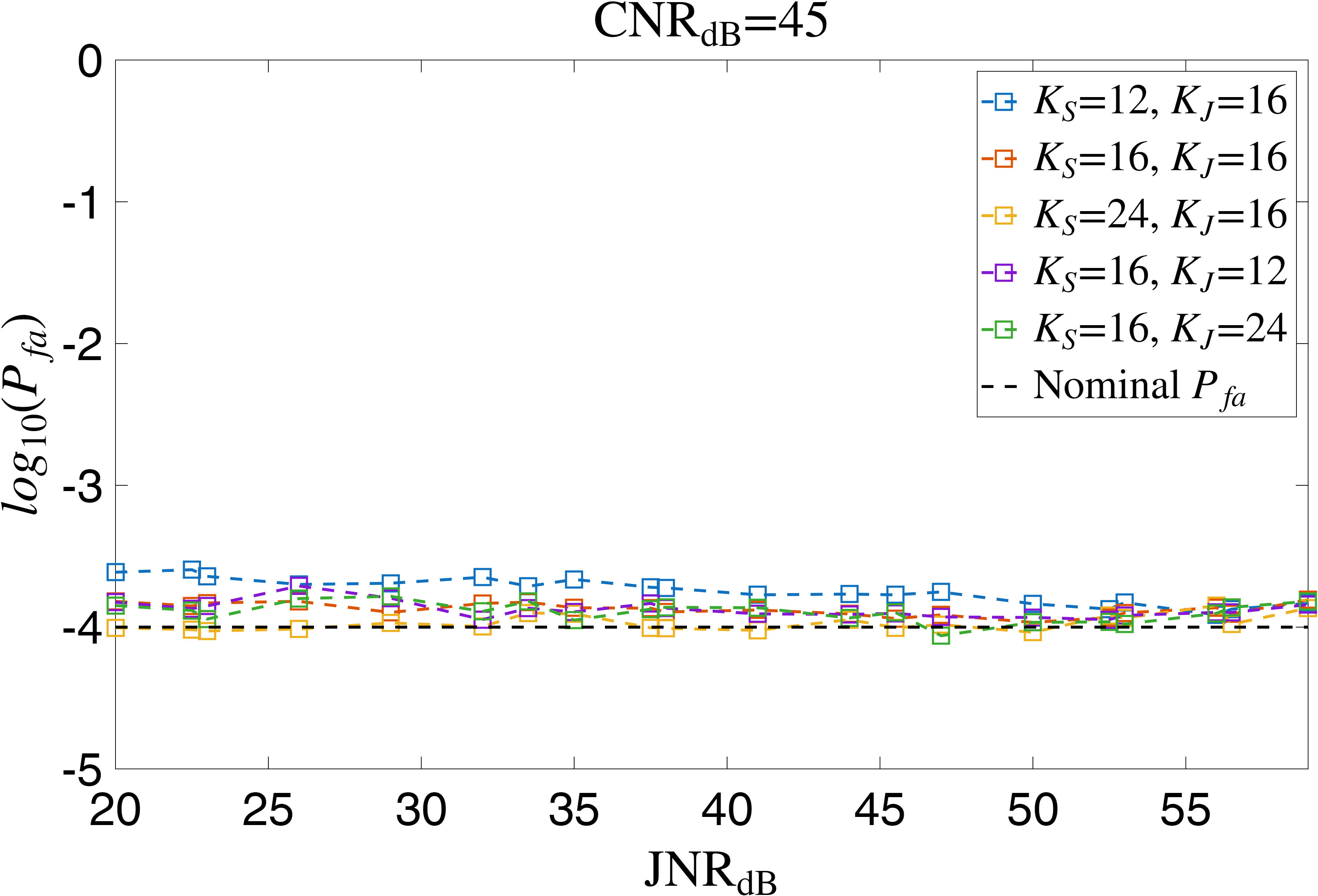}
        \caption{}
        \label{fig:PFA_adapt_CNR45}
    \end{subfigure}
    \hfill
    \begin{subfigure}[t]{0.32\textwidth}
        \centering
        \includegraphics[width=\linewidth]
        {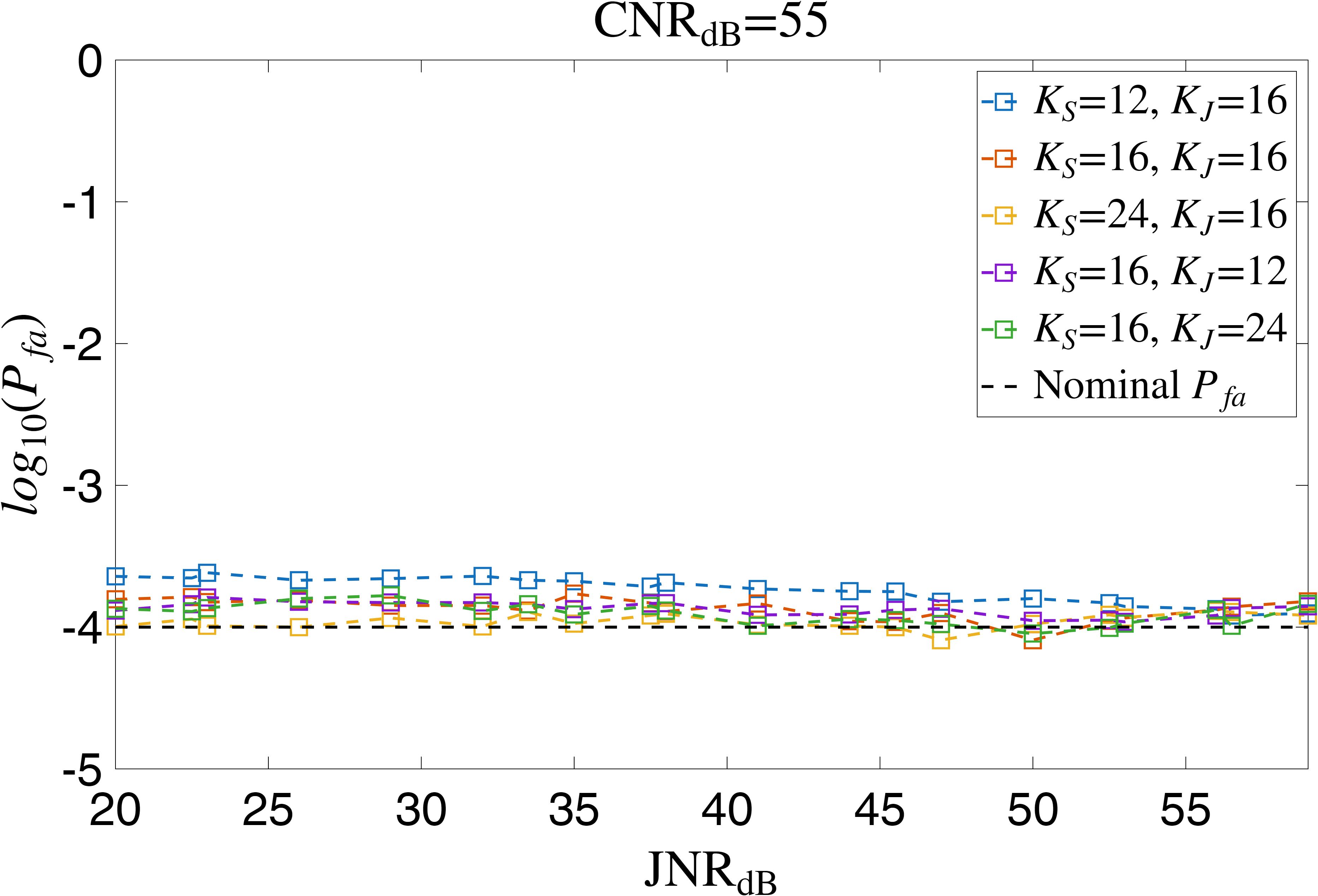}
        \caption{}
        \label{fig:PFA_adapt_CNR55}
    \end{subfigure}

    \caption{Actual $P_{fa}$ as a function of JNR obtained by using the 
    adaptive threshold-selection strategy for five data size configurations, namely
    $(K_P,K_S,K_J)=(16,12,16)$, $(16,16,16)$, $(16,24,16)$, $(16,16,12)$ and $(16,16,24)$, 
    and three different values of CNR: $35$ dB (a), $45$ dB (b), $55$ dB (c).}
    \label{fig:PFA_analysis_adapt_th}
\end{figure*}

\subsection{Detection Performance}
In this section, we investigate the detection capabilities of SOJ-LRT in
comparison with the considered competitor (\ref{eqn:Competitor}). 
As for the false-alarm analysis, $K_P$ is set to $16$, whereas 
$K_S$ and $K_J$ are independently varied over $\{12,16,24\}$ 
to assess the dependence of $P_d$ on the number of secondary ($\bZ_S$) and additional ($\bX$) data. 
As to the SOJ-LRT detector, in the following, we focus on its implementation 
based on the adaptive threshold-selection strategy, since it can be used
in practical scenarios without requiring knowledge of 
the actual CNR and JNR values.

Figure \ref{fig:Pd_vari_Ks} compares the SOJ-LRT with the competitor
for $\mbox{CNR}=30$ dB and $\mbox{JNR}=40$ dB. The figure assumes $K_P=16$ and $K_J=16$ and the $P_d$ curves
are compared by varying $K_S\in{[12, 16, 24]}$. 
It is worth highlighting that the competitor is not reported 
the sample starved scenario $K_S=12<N$ where its test statistics cannot be evaluated.
The observed results show that, for the considered parameter values 
the SOJ-LRT significantly outperforms the competitor. 
Furthermore, both detectors benefit from an increase in the number of 
secondary data, which enables a more accurate estimation 
of the clutter covariance matrix and shifts the detection curves 
toward lower SNR values. This improvement is more pronounced 
for the competitor; consequently, the performance gap 
between the two detectors decreases as $K_S$ increases. 
Nevertheless, the SOJ-LRT retains a detection advantage even 
for $K_S=24$. Finally, unlike the competitor, the SOJ-LRT can also 
operate for low volumes of secondary data.

\begin{figure}[t]
    \centering
    \includegraphics[width=\columnwidth]{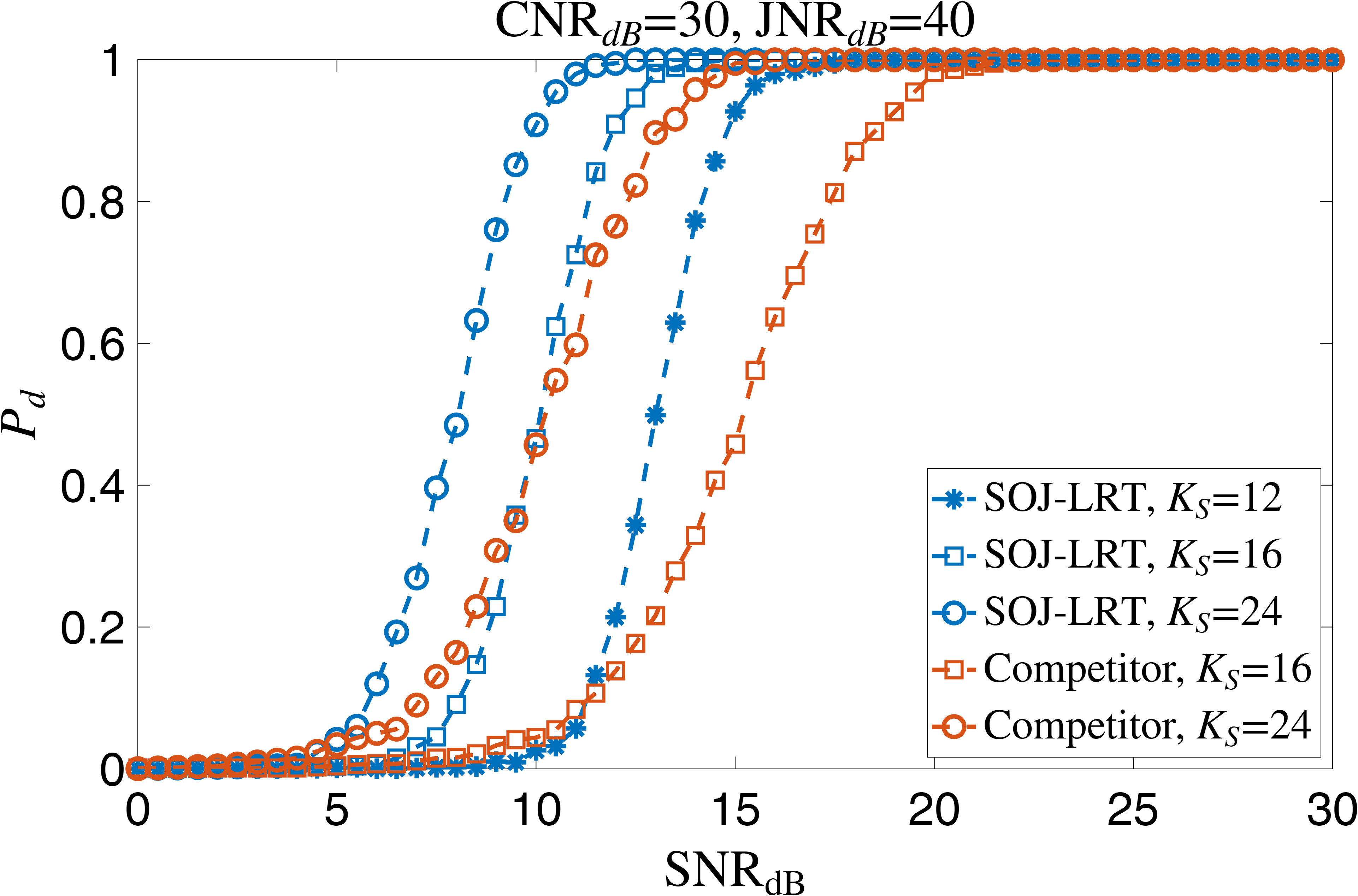}
    \caption{$P_d$ versus SNR for the SOJ-LRT and the considered competitor
    assuming different values of $K_S$, $\mbox{CNR}=30$ dB, $\mbox{JNR}=40$ dB, $K_P=16$, and $K_J=16$.}
    \label{fig:Pd_vari_Ks}
\end{figure}

Figure \ref{fig:Pd_vari_JNR} compares the two detection architectures
for three different JNR values, namely $30$, $40$, and $50$ dB. In all cases, 
the curves are obtained assuming $\mbox{CNR}=30$ dB, $K_P=16$, and $K_S=24$. 
For the SOJ-LRT, $K_J$ is set to $12$ in Figure \ref{fig:Vari_JNR_CNR_30_Kp_16_Ks_24_Kj_12}, 
to $16$ in Figure \ref{fig:Vari_JNR_CNR_30_Kp_16_Ks_24_Kj_16}, and to 
$24$ in Figure \ref{fig:Vari_JNR_CNR_30_Kp_16_Ks_24_Kj_24}. 
As for $K_S$, it is set to $24$ because such a value is
the most favorable for the competitor among those considered in this analysis, 
thus providing a conservative comparison with the SOJ-LRT.
The results in Figure \ref{fig:Pd_vari_JNR} show that the SOJ-LRT generally 
outperforms the competitor. However, we notice that 
the performance of SOJ-LRT decreases as JNR increases, whereas that 
of the competitor seems insensitive to the JNR. 
This behavior is explained by the projection of the primary data onto the orthogonal 
complement of the jammer subspace making the competitor essentially robust to 
the jammer power variations.  
Finally, we points out that SOJ-LRT benefits from the increment 
in the number of additional data ($K_J$). This behavior depends on 
the improved accuracy in the estimates of the noise power and the 
jammer covariance matrix that is obtained by increasing the 
number of additional data. This performance improvement is more 
significant at the higher JNR values $40$ and $50$ dB indicating that
a large number of additional data is required against strong jamming.

\begin{figure*}[t]
    \centering
    \begin{subfigure}[t]{0.32\textwidth}
        \centering
        \includegraphics[width=\linewidth]
        {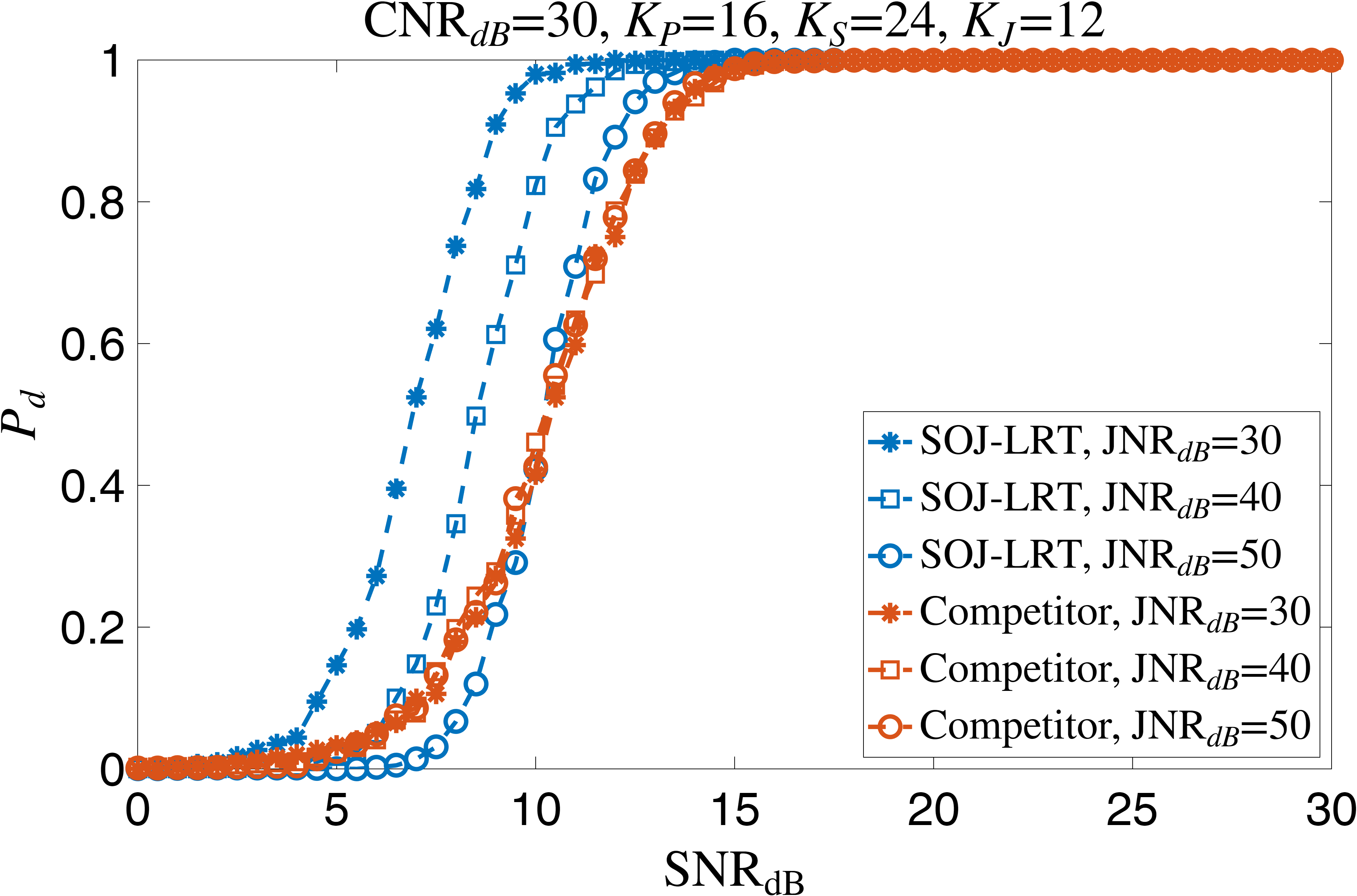}
        \caption{}
        \label{fig:Vari_JNR_CNR_30_Kp_16_Ks_24_Kj_12}
    \end{subfigure}
    \hfill
    \begin{subfigure}[t]{0.32\textwidth}
        \centering
        \includegraphics[width=\linewidth]
        {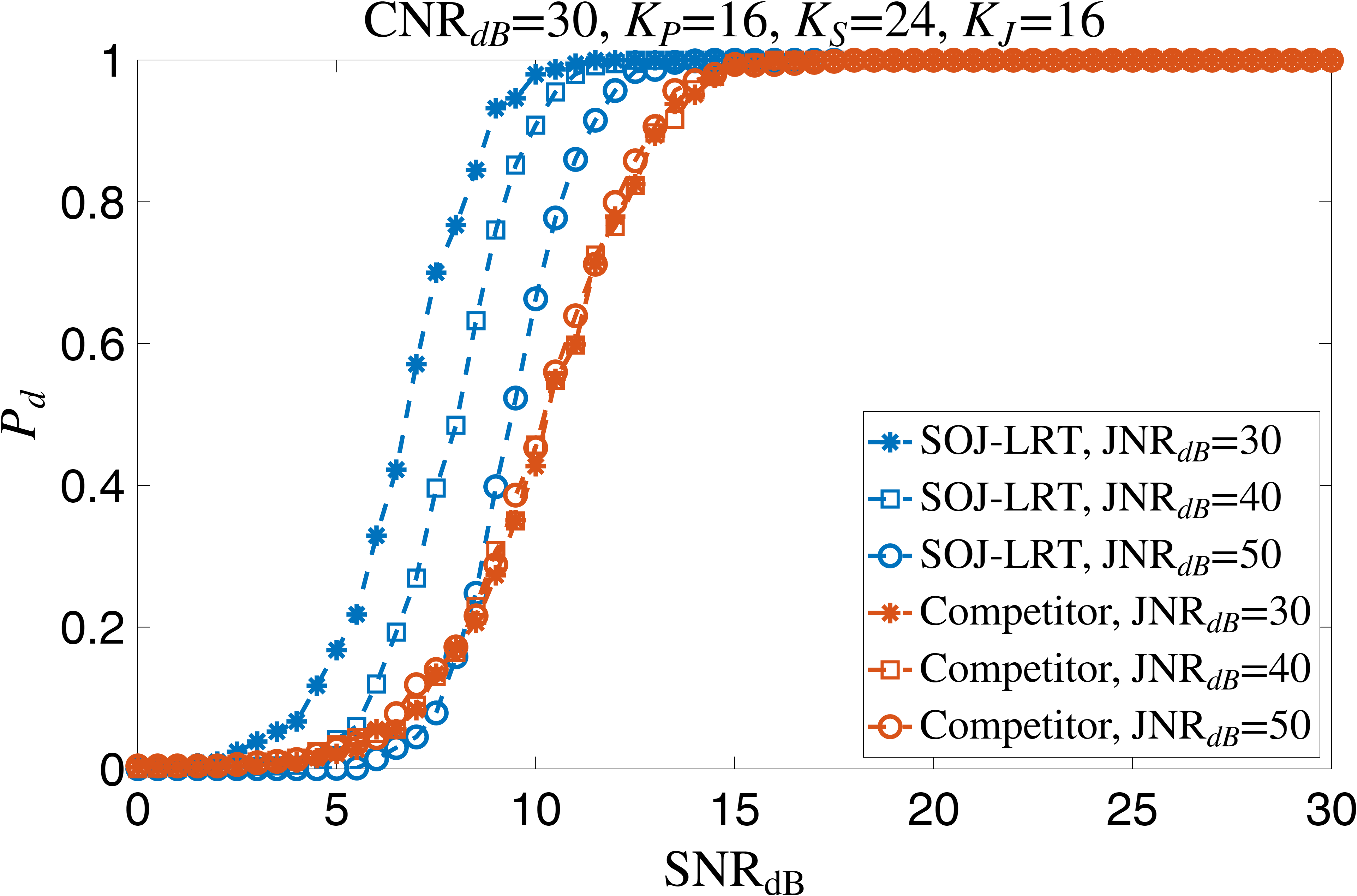}
        \caption{}
        \label{fig:Vari_JNR_CNR_30_Kp_16_Ks_24_Kj_16}
    \end{subfigure}
    \hfill
    \begin{subfigure}[t]{0.32\textwidth}
        \centering
        \includegraphics[width=\linewidth]
        {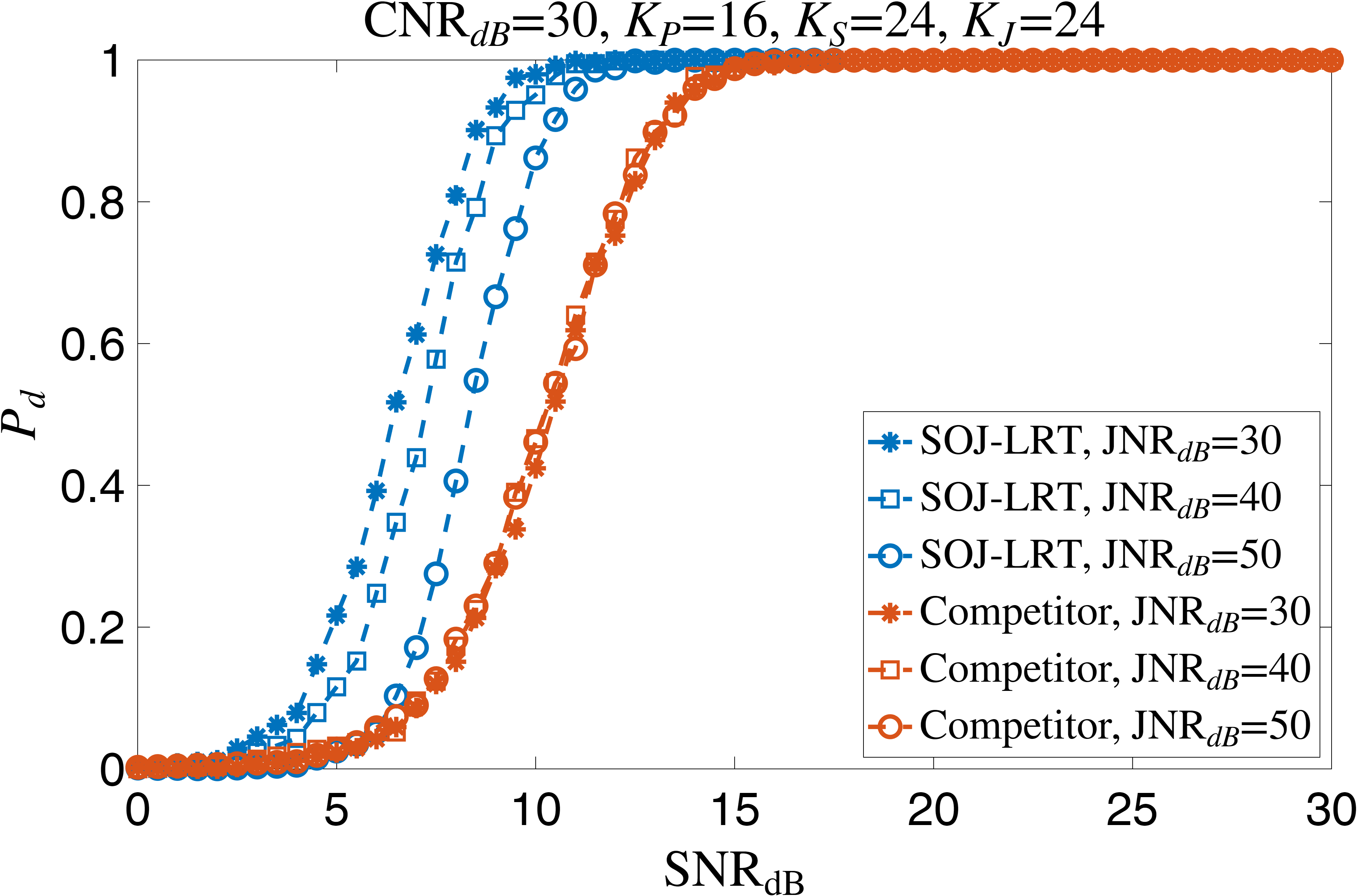}
        \caption{}
        \label{fig:Vari_JNR_CNR_30_Kp_16_Ks_24_Kj_24}
    \end{subfigure}

    \caption{Detection performance of the detectors for JNR$\in{[30,40,50]}$ $dB$, CNR=$30$ $dB$, $K_P=16$, $K_S=24$ and $K_J=12$ in (a), $K_J=16$ in (b), $K_J=24$ in (c).}
    \label{fig:Pd_vari_JNR}
\end{figure*}

\section{Conclusions}
\label{Sec:Conclusions}
This paper presented a subspace-based decision scheme to detect  range-spread targets  
in presence of thermal noise, clutter, and intentional interference. 
To the best of authors' knowledge, this is a first attempt to model the intentional interference, 
in addition to the target components, in terms of a second-order model, where 
both the target and jammer components are embedded into the covariance matrices of the received data.
As customary, secondary data are used to estimate the clutter covariance matrix;
an additional set of training
data, acquired by the system in passive mode, is used instead  to estimate
the jammer parameters. 
In fact, the jammer is present
in the secondary data vectors in addition to the clutter while data 
collected in passive mode are free of target and clutter components.
More importantly, the use of this additional data set is necessary
to solve an identifiability issue arising when the jammer and clutter covariance components 
are jointly estimated from the secondary data alone.

Based on this formulation, a detection scheme was derived exploiting
the LRT, where the unknown second-order parameters are suitably estimated by means
of a ``sequential'' maximum likelihood approach.
A unified theoretical framework was established to come up with 
closed-form second-order parameter estimates. This framework also provides a geometrical 
interpretation of the resulting estimation procedures.

Since
such architecture has a dependence of $P_{fa}$ on the JNR, an adaptive strategy for 
detection threshold setting is introduced so as to keep $P_{fa}$ approximately 
constant as the JNR varies. More precisely, we resorted to an adaptive 
threshold computed  by linearly interpolating 
the  LUT at the available values of the  estimated JNR. 
The performance assessment, conducted also in comparison to a natural (first-order) competitor,  
shows the effectiveness of the proposed solution that is able to 
outperform its competitor and works also under sample-starved scenarios.

Future research tracks might consider the extension of this approach
to heterogeneous environments where clutter in primary and secondary data
are characterized by different power levels. Another interesting point is the combination
of the first-order (representative of coherent jammers) and second-order jamming components.

\bibliographystyle{IEEEtran}
\bibliography{group_bib}

\end{document}

%% file: Defs2.tex
\newcommand{\bGamma}{\mbox{\boldmath{$\Gamma$}}}

\newcommand{\bbm}{\begin{bmatrix}}
\newcommand{\ebm}{\end{bmatrix}}

\newcommand{\bit}{\begin{itemize}}
\newcommand{\eit}{\end{itemize}}

\newcommand{\ben}{\begin{enumerate}}
\newcommand{\een}{\end{enumerate}}

\newcommand{\bdesc}{\begin{description}}
\newcommand{\edesc}{\end{description}}

\newcommand{\bea}{\begin{array}}
\newcommand{\eea}{\end{array}}

\newcommand{\tr}{\mbox{\rm Tr}\, }

\DeclareMathOperator{\sinc}{sinc}

\newcommand{\beqa}{\begin{eqnarray}}
\newcommand{\eeqa}{\end{eqnarray}}

\newcommand{\Comment}[1]{}

\newtheorem{prop}{Proposition}
\newtheorem{corollary}{Corollary}[prop]

\def\C{{\mathds C}}

\def\cC{\mbox{$\CMcal C$}}

\def\cL{\mbox{$\mathcal L$}}
\def\cN{\mbox{$\CMcal N$}}

\newcommand{\be}{\begin{equation}}

\newcommand{\ee}{\end{equation}}

\newcommand{\bzero}{{\mbox{\boldmath $0$}}}

\newcommand{\bm}{{\mbox{\boldmath $m$}}}

\newcommand{\bv}{{\mbox{\boldmath $v$}}}
\newcommand{\bx}{{\mbox{\boldmath $x$}}}

\newcommand{\bq}{{\mbox{\boldmath $q$}}}
\newcommand{\bz}{{\mbox{\boldmath $z$}}}

\newcommand{\bA}{{\mbox{\boldmath $A$}}}
\newcommand{\bB}{{\mbox{\boldmath $B$}}}
\newcommand{\bC}{{\mbox{\boldmath $C$}}}

\newcommand{\bE}{{\mbox{\boldmath $E$}}}
\newcommand{\bF}{{\mbox{\boldmath $F$}}}
\newcommand{\bG}{{\mbox{\boldmath $G$}}}
\newcommand{\bH}{{\mbox{\boldmath $H$}}}
\newcommand{\bI}{{\mbox{\boldmath $I$}}}
\newcommand{\bJ}{{\mbox{\boldmath $J$}}}

\newcommand{\bM}{{\mbox{\boldmath $M$}}}
\newcommand{\bL}{{\mbox{\boldmath $L$}}}

\newcommand{\bP}{{\mbox{\boldmath $P$}}}
\newcommand{\bQ}{{\mbox{\boldmath $Q$}}}
\newcommand{\bR}{{\mbox{\boldmath $R$}}}
\newcommand{\bS}{{\mbox{\boldmath $S$}}}
\newcommand{\bT}{{\mbox{\boldmath $T$}}}
\newcommand{\bU}{{\mbox{\boldmath $U$}}}
\newcommand{\bV}{{\mbox{\boldmath $V$}}}
\newcommand{\bX}{{\mbox{\boldmath $X$}}}
\newcommand{\bY}{{\mbox{\boldmath $Y$}}}
\newcommand{\bW}{{\mbox{\boldmath $W$}}}
\newcommand{\bZ}{{\mbox{\boldmath $Z$}}}

\newcommand{\diag}{\mbox{\rm diag}\, }

\newcommand{\bLambda}{{\mbox{\boldmath $\Lambda$}}}

\newcommand{\bSigma}{{\mbox{\boldmath $\Sigma$}}}

\newcommand{\test}{\mbox{$
\begin{array}{c}
\stackrel{ \stackrel{\textstyle H_1}{\textstyle >} }{
\stackrel{\textstyle <}{\textstyle H_0} }
\end{array}
$}}